\documentclass[journal,12pt,onecolumn,draftclsnofoot,]{IEEEtran}\IEEEoverridecommandlockouts
\usepackage{cite}
\usepackage{amsmath,amssymb,amsfonts}
\usepackage{graphicx}
\usepackage{textcomp}
\usepackage[dvipsnames]{xcolor}
\def\BibTeX{{\rm B\kern-.05em{\sc i\kern-.025em b}\kern-.08em
    T\kern-.1667em\lower.7ex\hbox{E}\kern-.125emX}}
    \usepackage{times}
\usepackage{lipsum}
\usepackage{mwe}
\usepackage{fancybox}
\usepackage{multicol}
\usepackage{color}
\usepackage{graphicx}
\usepackage{epsfig}
\usepackage{graphicx,latexsym}
\usepackage{amssymb,amsthm,amsmath}
\usepackage{array}
\usepackage{xy}

\usepackage{amsthm}
\usepackage{mathtools}
\usepackage{epstopdf}
\usepackage{algorithm}
\usepackage{algpseudocode}
\usepackage{algcompatible,lipsum}
\usepackage{cite}
\usepackage{ amssymb }
\usepackage{float}
\usepackage{multirow}

\usepackage{hhline}
\usepackage{url}
\usepackage{booktabs}
\usepackage{amsmath,algorithm,tabularx}
\usepackage{subcaption}
\usepackage{dblfloatfix}

\newcolumntype{C}[1]{>{\centering\arraybackslash$}p{#1}<{$}}
\makeatletter

\newcommand{\Rmnum}[1]{\expandafter\@slowromancap\romannumeral #1@}

\newtheorem{theorem}{Theorem}

\newtheorem{proposition}[theorem]{Proposition}

\newtheorem{remark}{Remark}

\newtheorem{lemma}[remark]{Lemma}
\newcommand{\multiline}[1]{%
  \begin{tabularx}{\dimexpr\linewidth-\ALG@thistlm}[t]{@{}X@{}}
    #1
  \end{tabularx}
}
\newcommand{\algmargin}{\the\ALG@thistlm}   
\makeatother
\algnewcommand{\parState}[1]{\State%
    \parbox[t]{\dimexpr\linewidth-\algmargin}{\strut #1\strut}}
\makeatother

\begin{document}

\title{User Dynamics-Aware Edge Caching and Computing for Mobile Virtual Reality}
\author{Mushu~Li,~\IEEEmembership{Member,~IEEE,}
Jie~Gao,~\IEEEmembership{Senior~Member,~IEEE,}
Conghao~Zhou,~\IEEEmembership{Member,~IEEE,}
Xuemin~(Sherman)~Shen,~\IEEEmembership{Fellow,~IEEE,} 
and~Weihua~Zhuang,~\IEEEmembership{Fellow,~IEEE} 
\thanks{
Mushu Li is with the Department of Electrical, Computer and Biomedical Engineering, Toronto Metropolitan University, Toronto, Canada M5B 2K3 (email: mushu.li@ieee.org).

Jie Gao is with the School of Information Technology, Carleton University, Ottawa, ON, Canada, K1S 5B6 (email: jie.gao6@carleton.ca).

Conghao Zhou, Xuemin (Sherman) Shen, and Weihua Zhuang are with the Department of Electrical and Computer Engineering, University of Waterloo, Waterloo, ON, Canada, N2L 3G1 (email: \{c89zhou, sshen,  wzhuang\}@uwaterloo.ca).    

}
}%

\maketitle

\begin{abstract}

In this paper, we present a novel content caching and delivery approach for mobile virtual reality (VR) video streaming. The proposed approach aims to maximize VR video streaming performance, i.e., minimizing video frame missing rate, by proactively caching popular VR video chunks and adaptively scheduling computing resources at an edge server based on user and network dynamics. First, we design a scalable content placement scheme for deciding which video chunks to cache at the edge server based on tradeoffs between computing and caching resource consumption. Second, we propose a machine learning-assisted VR video delivery scheme, which allocates computing resources at the edge server to satisfy video delivery requests from multiple VR headsets. A Whittle index-based method is adopted to reduce the video frame missing rate by identifying network and user dynamics with low signaling overhead. Simulation results demonstrate that the proposed approach can significantly improve VR video streaming performance over conventional caching and computing resource scheduling strategies.
\end{abstract}


\section{Introduction}
As a promising use case in sixth-generation (6G) communication networks, mobile virtual reality (VR) is anticipated to reshape the way people study, work, and entertain via the digital transformation of the physical world~\cite{9830046}. A myriad of novel applications can be enabled by mobile VR, such as immersive gaming, telesurgery, metaverses, etc.~\cite{review}. With mobile VR, users can watch 360\textdegree ~stereoscopic videos through VR headsets and access the digital world with a fully immersive experience. 

Delivering VR videos to mobile VR headsets in real time can be challenging. VR videos are characterized by their ultra-high resolutions (up to 11520 $\times$ 6480 pixels for panoramic videos), while conventional videos typically have a resolution of 4K or less~\cite{Pantelis}. From a communication perspective, VR video streaming consumes considerable bandwidth on both backhaul and wireless links for delivering high-resolution videos to VR headsets. {In addition, VR headsets play 360\textdegree ~stereoscopic videos, i.e., two different videos for the left and right eyes, respectively, creating the perception of depth that people experience in the real world. To reduce storage resource usage, monoscopic videos, which consist of flat two-dimensional equirectangular video frames, can be stored at the cloud server and projected into stereoscopic videos in content delivery~\cite{xie2016deep3d}.} Extensive VR video processing requires powerful computing capabilities, and computing capabilities in VR headsets are usually insufficient for low-latency video processing~\cite{Boos}.

The aforementioned challenges in VR video delivery can be addressed using two main approaches: tile-based content delivery and mobile edge computing. First, in tile-based content delivery, VR videos are divided spatio-temparlly into small video chunks (VCs).  Only the VCs that cover the user's current viewing area, i.e., the field of view (FoV), are sent to and played on the VR headset~\cite{Dai}. {Such a tile-based solution can significantly reduce the data sizes of videos to be downloaded from the cloud server while satisfying users' requests.} Second, mobile edge computing (MEC), as an innovation in the fifth-generation (5G) networks, leverages storage and computing resources at access points near VR headsets to reduce content delivery delay. Edge servers are located at access points to cache popular VCs and {project monoscopic VCs (MVCs) into stereoscopic VCs (SVCs)} for resource-limited VR headsets~\cite{Liu}. While MEC and tile-based content delivery make VR video streaming possible, the following questions need to be answered:
\begin{itemize}
    \item  Given limited storage space, which VCs should be cached at an edge server?
    \item How should computing resources at the edge server be scheduled to meet video delivery requests from multiple VR headsets?
\end{itemize}

Addressing the first question necessities the effective selection of VCs cached at an edge server, referred to as content placement. Determining the optimal caching solution can be complicated ~\cite{gao2020design}. {The edge server can cache MVCs, which have smaller data sizes than the corresponding SVCs covering the same spatial area, to reduce caching resource usage given limited cache size. Alternatively, the edge server can cache SVCs, which are processed VCs and ready to send to VR headsets, to reduce video processing time in content delivery~\cite{8728029}.
Furthermore, spatio-temporally partitioning a short video can produce a large number of VCs. These VCs may be spatially overlapped and stitched together to produce the video covering a user's FoV~\cite{Dai}. The vast number of potential caching options and coupling among VCs complicate content placement at the edge server. }

Regarding the second question, managing limited computing resources for satisfying multiple content delivery requests is non-trivial due to the dynamics in VR video streaming. Millimeter wave (mmWave) communications are widely adopted to enable high-speed VR content delivery in many research works~\cite{yang2022feeling}. However, physical obstacles in the propagation environment, including the user's body, can disrupt communication links, resulting in a time-varying transmission rate. In addition, proactive content delivery can facilitate smooth VR video streaming by requesting and downloading VCs in advance based on predicted user viewing trajectories~\cite{8918038}. {The spatial area in a video displayed to a user changes over time as the user's head orientation changes, and the viewing trajectories can be different for different videos depending on the content (e.g., sports or movies) and for different users due to the user viewing preference. The dynamic of user viewing trajectories} introduces additional uncertainty in computing resource management at the edge server \cite{li2022adaptive}. Allocating computing and communication resources for satisfying VR content delivery requests based on inaccurate viewing trajectory predictions may degrade resource utilization and result in a poor VR video streaming experience for all nearby VR headsets sharing the resources.

In this paper, we study content caching and delivery for mobile VR video streaming. The research objective is to optimize VR video streaming performance, i.e., reducing the frame missing rate, for multiple VR users given limited communication, caching, and computing resources. First, we propose a content placement scheme to determine the optimal caching solution, given a content delivery delay threshold, channel quality and processing rate statistics, and VC data sizes. The proposed heuristic scheme uses a distributed optimization technique, i.e., alternating direction method of multipliers (ADMM), to facilitate parallel content placement. Using the caching solution, we propose a machine learning-based scheduling scheme to allocate limited computing resources at the edge server to reduce the frame missing rate in satisfying content delivery requests generated by multiple VR headsets. Finally, we study adaptive video quality adjustment to maximize the cached video quality at the edge server. This study covers the main aspects of VR video streaming, and our contributions are summarized as follows:

\begin{itemize}
    \item We develop a caching placement and computing resource management approach to effectively improve video streaming performance in mobile VR video streaming.
\item We propose a content placement scheme for caching VCs by trading off caching resource utilization and content delivery delay. The proposed scheme is proven to achieve a near-optimal caching solution. 
    \item We propose a deep reinforcement learning (DRL)-based content delivery scheduling scheme for computing resource management. The scheme minimizes the frame missing rate in video streaming by adapting to network and user dynamics in a distributed manner.

\end{itemize}

The remainder of this paper is organized as follows. Section II introduces related works. In Sections III and IV, we present the system model and problem formulation, respectively. In Sections V and VI, we propose content caching and content delivery scheduling schemes, respectively. We then discuss adaptive video quality adjustment in Section VI. Simulation results are presented in Section VII. Section VIII concludes this study.

\section{Related Works}

{MEC can reduce content delivery and processing delay in VR video streaming by providing storage and computing capabilities near mobile VR headsets. An edge server can cache both MVCs and SVCs. 
A caching solution was evaluated  in~\cite{8728029} to minimize average content delivery delay based on the tradeoffs among communication, computing, and caching resource usage. A similar problem was investigated in~\cite{8713498} considering caching capabilities at VR headsets. A mixed-integer problem was formulated and solved to determine whether and where to cache MVCs and SVCs.
{Additionally, some works take into account the correlations among VCs in order to reduce resource usage in caching and delivering videos covering users' FoVs, and the correlations vary depending on the granularity of partitioning a VR video.} When a VR video is spatially partitioned into small segments, multiple VCs need to be stitched together to generate the video covering a user's FoV. In such a case, a tile-based content placement solution was proposed in~\cite{9350227}.  A Choquet integral method was applied to identify the popularity of individual VCs based on user viewing behavior statistics, and popular VCs were cached at the edge server.
The estimation of VC popularity was also studied in~\cite{9343267} and~\cite{papaioannou2019tile} for caching popular and high-quality VCs, while minimizing content delivery delay.
When a video is coarsely partitioned in the spatial dimension, a single VC can cover a user's FoV, whereas multiple VCs may cover overlapped spatial areas.
A study of content placement in such a case was presented in~\cite{Dai}, where a group of cached VCs are synthesized to generate VCs that are not cached yet. The objective was to minimize content delivery delay, while maximizing caching resource utilization by properly selecting VCs to cache. 
From the existing works, VR content placement is a complex problem with different types of VCs (different quality levels, monoscopic or stereoscopic, etc.), and a scalable content placement solution is necessary for effectively caching VCs given a large number of potential caching options. }

In addition to content placement, real-time content delivery scheduling and computing resource management can further improve video streaming performance for maximizing computing resource utilization at the edge server. To prevent users from feeling spatially disoriented and dizzy in VR video streaming, the time difference between the user's movement and corresponding display at the VR headset, referred to as motion-to-photon (MTP) delay, should be less than 20 ms~\cite{oculus}. 
In order to meet this strict requirement, a VR headset can predict user viewpoints based on historical viewpoint trajectories and request the VCs proactively~\cite{10.1145/3210240.3210323, Nasrabadi}. As a result, the VR video streaming experience depends on viewpoint prediction accuracy. In~\cite{9268977}, the viewpoint prediction window was determined to maximize users' VR experience, and computing resource allocation was jointly determined. Furthermore, machine learning algorithms were adopted in~\cite{8955928} and~\cite{9411714} to determine computing resource scheduling policies, given user viewpoint trajectories. In spite of the research efforts, designing a lightweight scheduling scheme remains an open issue when network dynamics and viewpoint prediction uncertainty are taken into account~\cite{li2022adaptive}.

\section{System Model}
\begin{figure}[t]
		\centering
	  	\includegraphics[width=0.5 \textwidth]{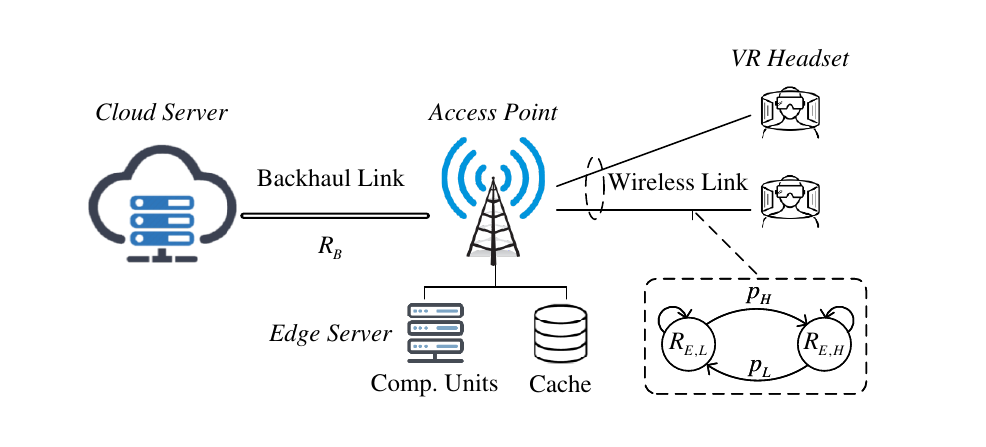}
	  	\caption{Network model.}\label{fig:1}
\end{figure}

The considered VR video caching and delivery scenario is shown in Fig.~\ref{fig:1}. An edge server, located at the access point, connects and serves $U$ mobile VR users. 
{Each user has a VR headset that plays 360\textdegree~3D stereoscopic VR videos at $\hat{f}$ frames per second (fps). We refer to the time interval between the playback of two adjacent frames as a time slot, and the length of a time slot is $\delta = 1/\hat{f}$.} 
In each time slot, each VR headset predicts the location of future user viewpoints, generates a content delivery request for the VC to be played, and sends the request to the edge server.  Video processing, including stitching multiple VCs and projecting 2D monoscopic videos to 3D stereoscopic videos, can consume extensive computing resources and is conducted at the edge server.

The edge server is equipped with $E$ computing units, and each computing unit can {conduct video processing for} only one content delivery request at any time. Moreover, the edge server has a cache with size $C$ for storing popular VCs. For retrieving VCs that are not cached, the edge server connects to the cloud server via a wired link with data rate $R_B$. When transmitting the requested VC from the edge server to the corresponding VR headset, a mmWave band is used with high-speed rate $R_H$, and a sub-6 GHz band acts as a backup with low-speed rate $R_L$ if the high-speed mmWave band is in outage~\cite{Liu}. We use a two-stage Markov chain to model the data rate between the edge server and a VR headset. The probabilities of the high-speed state transiting from and to the low-speed state are denoted by $p_L$ and $p_H$, respectively. As a result, the average data rate from the edge server to a VR headset is $\bar{R}_E = (R_L +R_H)/{(p_L+p_H)}$. 
\subsection{VR Video Model}

{Tile-based content caching and delivery are considered. 
As shown in Fig.~\ref{fig:3}, each VR video has two dimensions. In the spatial dimension, a video can be divided into $I$ tiles, and we denote the number of tiles needed to cover an FoV by $\hat{I}$. In the temporal dimension, videos can be divided into $J$ time segments {with a constant duration in seconds\footnote{{As time slots are measured in milliseconds (e.g., for a 60 fps video, a time slot lasts 17 milliseconds), a time segment contains hundreds of time slots.}}~\cite{hooft2019tile}.} Accordingly, VR videos are divided into VCs along both dimensions. In the temporal dimension, the playback time of a VC is the duration of a time segment. }

In the spatial dimension, there are two types of VCs: MVCs and SVCs. An MVC contains the monographic content of a spatial tile. For MVCs, layered encoding is adopted to differentiate the video quality~\cite{Pantelis}. There are two types of MVCs: layer-0 MVCs, which have a low resolution to satisfy the minimum video quality requirement, and layer-1 MVCs, which can be added on top of layer-0 MVCs for quality improvement. We denote the sets of layer-0 and layer-1 MVCs for all VR videos by $\mathcal{T}_0$ and $\mathcal{T}_1$, respectively. Tuple $t = (i,j,y)$ represents the MVC corresponding to the $i$-th tile, $j$-th time segment, and $y$-th layer, where $i\in \{1,\dots,I\}$, $j\in \{1,\dots,J\}$, and $y\in \{1,0\}$. 

There are $X$ MVCs, including both layer-0 and layer-1 MVCs, to be stitched and projected to an SVC, and the projection is facilitated by using auxiliary files (e.g., depth map). 
SVCs are delivered by the edge server to mobile VR headsets. To obtain the SVC centering on tile $i$, $X_{0}$ layer-0 MVCs and $X-X_{0}$ layer-1 MVCs, corresponding to the tiles closest to tile $i$, are stitched and projected. In the spatial dimension, the area covered by an SVC depends on the parameter $X_0$, which must be at least $\hat{I}$ to ensure that an SVC can fully cover the area of an FoV. The video quality of an SVC can be enhanced by encoding a large portion of layer-1 MVCs~\cite{9939105, 8926488}. On the other hand, a VR headset can play the videos covering multiple FoVs when an SVC is encoded with a large portion of layer-0 MVCs (corresponding to a large $X_0$), enabling the VR headset to adapt to frequent viewpoint changes~\cite{shi2019mobile}.   
  
Denote the set of all SVCs by $\mathcal{T}_S$, and denote the sets of layer-0 and layer-1 MVCs for generating SVC $f$ by $\mathcal{B}(f, X_0)$ and $\mathcal{E}(f, X_0)$, respectively, where $\mathcal{B}(f, X_0) \subseteq \mathcal{T}_0$ and $\mathcal{E}(f, X_0) \subseteq \mathcal{T}_1$. Since an SVC contains videos for both eyes, the overall data size of all MVCs used to generate an SVC is smaller than the data size of the SVC. Therefore, caching SVCs consumes more cache spaces at the edge server but decreases video processing delay when delivering SVCs to VR headsets. In the considered scenario, the cloud server stores all MVCs, and the edge server can cache some MVCs and SVCs to reduce content delivery delay. The sets of MVCs and SVCs cached at the edge server are denoted by $\mathcal{M}$ and $\mathcal{S}$, respectively.

\begin{figure}[t]
		\centering
	  	\includegraphics[width=0.5 \textwidth]{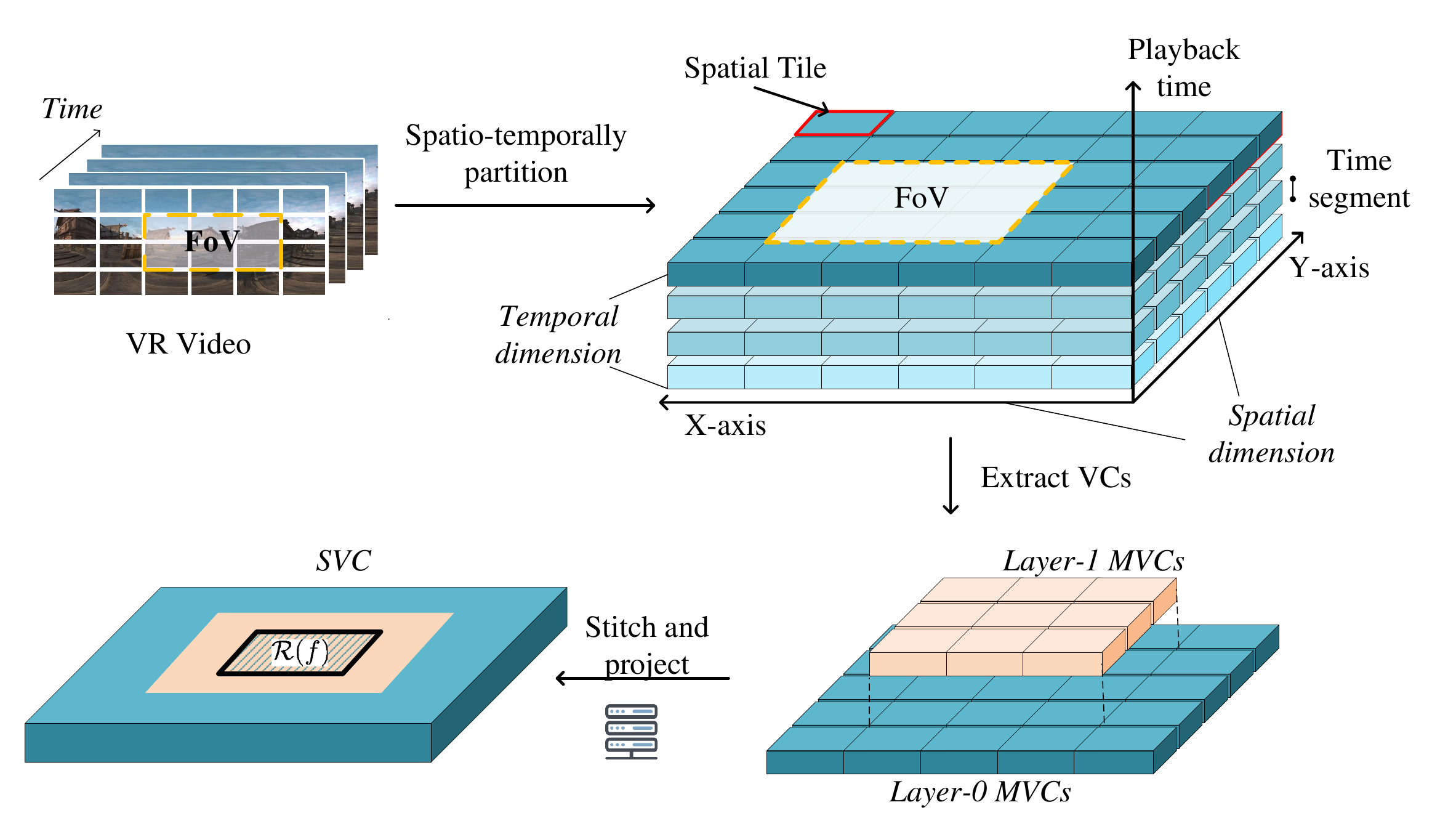}
	  	\caption{VR video model.}\label{fig:3}
\end{figure}

\subsection{Content Delivery Model}\label{model1}

The VR content delivery procedure is illustrated in Fig.~\ref{fig:2}. Each VR headset has a local buffer to store the downloaded SVCs. In each time slot, a VR headset tracks the user viewpoint and {plays} the SVC covering the user's FoV. Denote the viewpoint of user $u$ in time slot $k$ by $v_{u,k} = (i_{u,k}, j_{u,k})$, where $i_{u,k}$ and $j_{u,k}$ represent {the spatial tile at the center of user $u$'s FoV in time slot $k$ and the time segment of the video played by the VR headset in time slot $k$, respectively}. When an SVC covers a larger area than an FoV, i.e., $X_0>\hat{I}$, multiple viewpoints can be rendered by an SVC. Denote the set of viewpoints that can be rendered by SVC $f$ by $\mathcal{R}_f$. In time slot $k$, if SVC $f$ in the buffer of VR headset $u$ can render the current viewpoint, i.e., $v_{u,k} \in \mathcal{R}_f$, the SVC is fetched and played by the VR headset. 
Otherwise, if no SVC in the buffer can render the viewpoint, the VR headset fails to play the correct video frame, and an frame missing event occurs. {In addition, to identify which SVC will be requested proactively, in each time slot, each VR headset consecutively predicts user viewpoints to be rendered after the current slot, until it finds the next future viewpoint that cannot be rendered by any SVCs in the current local buffer.} Such a predicted viewpoint is referred to as the desired viewpoint. 
Then, the VR headset generates a VR content delivery request proactively for obtaining the SVC to render the desired viewpoint. Denote the predicted viewpoint of time slot $k$ for VR headset $u$ by $\hat{v}_{u,k}$. The content delivery request from VR headset $u$ in time slot $k$ is denoted by $d_{u,k} = \hat{v}_{u,k'}|_{k'>k}$, where $k'$ is the time slot to render desired viewpoint $\hat{v}_{u,k'}$. {Since the historical information of the past viewpoint trajectories is updated after each time slot, a VR headset conducts predictions for future viewpoints in each time slot. Only the prediction result obtained in the current time slot is used to generate the content delivery request.}

The edge server receives content delivery requests from multiple VR headsets in each time slot. A scheduler at the edge server identifies which request to respond first when a computing unit is available. If a computing unit is allocated to a request, i.e., the request is scheduled, there are two possible scenarios for fetching and processing the VCs to satisfy the request:
\begin{enumerate}
    \item Case 1: The edge server caches an SVC that can satisfy the request, i.e., $\{\exists f| f\in \mathcal{S}, d_{u,k} \in \mathcal{R}_f\}$. In this case, the SVC is fetched from the cache and delivered to the corresponding VR headset\footnote{In Case 1, a computing unit is still needed to process the SVC to satisfy the video requirements of the corresponding VR headset~\cite{8926340}, although the processing delay is negligible.}. 
    \item Case 2: No SVC cached on the edge server can satisfy the request. In this case, the edge server fetches MVCs in $\mathcal{B}(f,X_0)$ and $\mathcal{E}(f, X_0)$ from the cache or the cloud server, processes the MVCs for generating SVC $f$, and delivers SVC $f$ to the corresponding VR headset.
\end{enumerate}
Regardless of whether the request from a VR headset is scheduled or not, the VR headset tracks and predicts the viewpoints in each time slot. If the viewpoint prediction result changes before the request is scheduled, the VR headset sends an updated content delivery request to the edge server, and the previously unscheduled request becomes obsolete and discarded. 

\begin{figure}[t]
		\centering
	  	\includegraphics[width=0.5\textwidth]{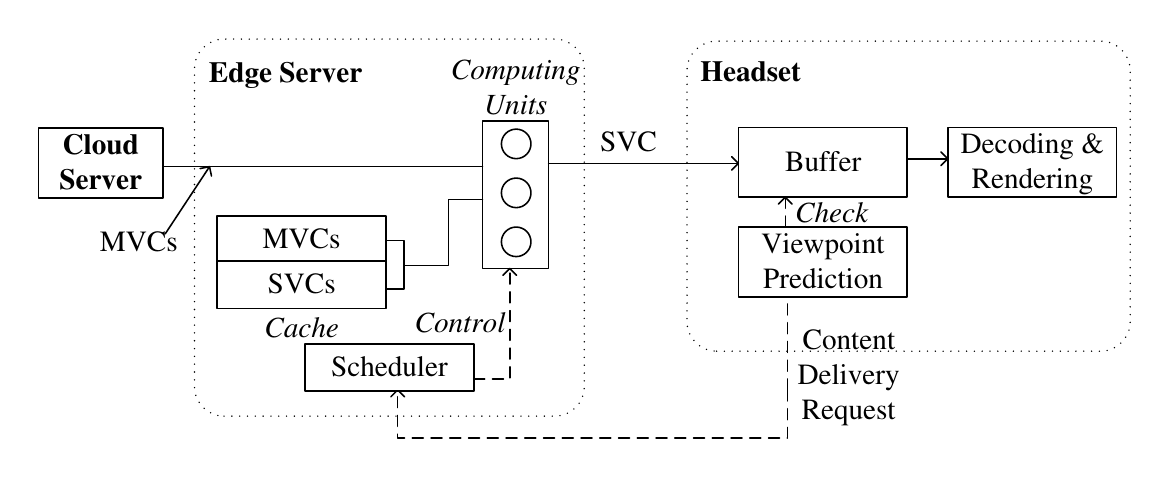}
	  	\caption{Content delivery procedure.}\label{fig:2}
\end{figure}

\subsection{Delay Model}
A content delivery delay is the duration from the instant that a request is scheduled to the instant that the SVC rendering the desired viewpoint is downloaded to the VR headset. The delay includes three parts: the transmission delay for downloading MVCs that are not cached from the cloud server to the edge server, the computing delay for stitching and processing the MVCs, and the transmission delay for downloading the SVC from the edge server to the VR headset. When the request from VR headset $u$ in time slot $k$, i.e., request $d_{u,k}$, is scheduled, a computing unit is allocated to process the request and obtain the SVC. 
The delay for transmitting the SVC from the edge server to the VR headset is 
\begin{equation}
    T^{\text{E}}(d_{u,k}; R) = {w_{f}^\text{S}}{R}^{-1},  d_{u,k} \in \mathcal{R}_f
\end{equation}
where the value of $R$ can be either $R_{L}$ or $R_{H}$. Parameter $w_{f}^\text{S}$ represents the data size of SVC $f$, which can be approximated by the overall data size of MVCs in set $\mathcal{C}(f) = \mathcal{B}(f,X_0)\cup\mathcal{E}(f,X_0)$. Let $w_t^\text{M}$ denote the data size of MVC $t$. Then,
\begin{equation}
    w_{f}^\text{S} = \sum_{t\in \mathcal{C}(f)}\alpha w_t^\text{M}
\end{equation}
where parameter $\alpha$ represents the ratio of data size change from MVCs  to the corresponding SVC. If no cached SVC can satisfy the request, i.e., Case 2 in Subsection~\ref{model1}, the edge server needs to generate SVC $f$ from MVCs. Firstly, the MVCs not in the cache are downloaded from the cloud server. The transmission delay is 
\begin{equation}
    T^{\text{B}}(d_{u,k}) = \sum_{t \in \{\mathcal{C}(f)\backslash \{\mathcal{C}(f) \cap \mathcal{M}\}\}} w_{t}^\text{M} R_B^{-1},  d_{u,k} \in \mathcal{R}_f.
\end{equation}
When all MVCs in set $\mathcal{C}(f)$ are available at the edge server, the MVCs are stitched and projected to SVC $f$. 
The corresponding computing delay is
\begin{equation}
    T^{\text{C}}(d_{u,k}) = \sum_{t \in \mathcal{C}({f})} w_{t}^\text{M} \chi^{-1},   d_{u,k} \in \mathcal{R}_f
\end{equation}
where parameter $\chi$ denotes the time to process each data bit in MVC stitching and monoscopic-to-stereoscopic projection. 
Therefore, the overall delay of content delivery for request $d_{u,k}$ is
\begin{equation}
    T(d_{u,k}) = T^{\text{B}}(d_{u,k})+T^{\text{C}}(d_{u,k})+T^{\text{E}}(d_{u,k}, R).
\end{equation}

\section{Problem Formulation}
Our main objective is to minimize the frame-missing rate when playing VR videos on multiple VR headsets. Achieving the objective requires that the edge server efficiently caches VCs to reduce the average content delivery delay and dynamically allocate computing resources to individual content delivery requests in real time. Therefore, we formulate two optimization problems for caching content placement and content delivery scheduling, respectively.

\subsection{Content Placement}

The objective of content placement is to maximize the probability that content delivery delay $T(d_{u,k})$ is lower than the delay requirement $H$, i.e., $P(T(d_{u,k})<H)$, by selecting proper VCs to cache. The corresponding problem can be formulated as:
\begin{subequations}
\label{obj.caching}
\begin{align}
 \max_{\{\mathcal{M}, \mathcal{S}\}} &&& \sum_{d\in \mathcal{V}} p(d) P(T(d) < H) \label{obj.cachinga}\\
 \text{s.t.} &&& \sum_{t \in \{\mathcal{M} \}} w_t^\text{M} + \sum_{f \in \mathcal{S}} w_f^\text{S} \leq C ,\\
&&& \mathcal{M}\subseteq \{\mathcal{T}_0 \cup \mathcal{T}_1\}, \mathcal{S}\subseteq \mathcal{T}_S
\end{align}
\end{subequations}
where set $\mathcal{V}$ denotes all possible viewpoints. In \eqref{obj.cachinga}, $p(d)$ represents the probability that request $d$ is generated by a VR headset. Such probability can be obtained from historical viewpoint popularity profiles, which reflect the average fraction of time that a user viewpoint falls onto a specific tile in a specific time segment~\cite{wu2017dataset}. The problem is a non-linear combinatorial optimization problem since multiple viewpoints may be rendered by an SVC.  It is similar to the maximal coverage problem~\cite{khuller1999budgeted}, which is proved to be NP-hard. Furthermore, the number of VCs is large, which makes the problem more complex. We propose a two-stage solution for solving the content caching problem. We first present a heuristic algorithm to solve~\eqref{obj.caching} given a fixed cache size. Then, we adopt an iterative optimization method to enable parallel and scalable content placement.
\subsection{Content Delivery}
Due to limited computing resources, the scheduler at the edge server should allocate the computing resource to satisfy content delivery requests that can reduce the frame missing rate as much as possible. 
We refer to the probability that an SVC in the buffer matches the current viewpoint as the hit probability. A higher hit probability means a lower frame missing rate at a VR headset. 
Define scheduling variable $a_{u,k}$, where $a_{u,k}=1$ if the request from VR headset $u$ is scheduled in time slot $k$, and $a_{u,k} = 0$ otherwise. 
The scheduling variables are determined to maximize the hit probability, i.e., minimize the time-averaged number of events where SVCs in the local buffers of VR headsets can render the viewpoints of the corresponding users. The objective function is
\begin{align}
     \max_{a_{u,k}, \forall u,k} & \lim_{K\rightarrow \infty} \frac{1}{K} \mathbb{E} \Big[\sum_{u = 1}^{U} \sum_{k=1}^{K} \textbf{1} [v_{u,k} \in \mathcal{U}_{u,k}(\mathbf{a}_{u,k-1})] \Big] \label{eq.problem}
\end{align}
where $\mathcal{U}_{u,k}(\mathbf{a}_{u,k-1})$ denotes the set of viewpoints that can be rendered by SVCs in VR headset $u$'s buffer in time slot $k$. Such a set depends on the scheduling decisions $\mathbf{a}_{u,k-1} = [a_{u,1}, \dots, a_{u,k-1}]$. Function $\textbf{1} [x]$ is equal to one if $x$ is true and zero otherwise.
Due to the limited number of computing units, we formulate a stochastic constraint:
\begin{align}
\lim_{K\rightarrow \infty} \frac{1}{K}  \mathbb{E} \Big[\sum_{u = 1}^{U} \sum_{k=1}^{K} a_{u,k}\Big] \leq \frac{E\times \delta}{\bar{T}} \label{eq.constraint}
\end{align}
where $\bar{T}$ represents the average content delivery delay, i.e., $\bar{T} = \sum_{d\in \mathcal{V}}p(d)T(d)$. The average number of time slots for delivering an SVC is $\bar{T}/\delta$. The constraint in~\eqref{eq.constraint} is a relaxed constraint imposing that, on average, at most ${(E\times \delta)}/{\bar{T}}$ computing units can be scheduled in a time slot. 

Several factors affect the scheduling decision. Generally, the edge server should first schedule the request that needs to be satisfied urgently. However, if the channel quality is poor or receiving the SVC from the edge server takes a long time, the VR headset may not receive the requested SVC on time. In addition, high viewpoint movement dynamics pose a challenge to accurate viewpoint prediction. As a result, downloaded SVCs may not be played, resulting in a waste of computing and communication resources. 
Both network and user viewpoint movement dynamics should be taken into account in computing resource management. 

\section{Mobile VR Content Caching}\label{sec.caching}

In this section, we first propose a content placement scheme to solve Problem~\eqref{obj.caching} given {the predefined video quality setting $X_0$}, while taking into account the trade-offs among multiple resources. Then, we extend the scheme for caching a large number of VCs, and solve Problem~\eqref{obj.caching} in a parallel manner.
\subsection{Content Placement Scheme}
We first identify the properties of content placement for different types of VCs. Let $\mathcal{F} = \mathcal{M} \cup \mathcal{S}$.
Let function $L(\mathcal{F})$ represent the overall probability of delay requirement satisfaction as follows:
\begin{equation}
L(\mathcal{F}) =  \sum_{d \in \mathcal{R}(\mathcal{F})}p(d) P(T(d) < H) \label{eq.set}
\end{equation}
where $\mathcal{R}(\mathcal{F})$ represents the set of viewpoints that can be rendered by the VCs in set $\mathcal{F}$, and $\mathcal{R}(\mathcal{F})= \{\cup_{f\in \mathcal{S}}\mathcal{R}_f\} \cup  \{\cup_{f\in \{{f|\mathcal{C}(f)\subseteq \mathcal{M}}\}}\mathcal{R}_f\}$.
The following lemmas can be obtained.
\begin{lemma}\label{le.1}
If the edge server caches only SVCs, i.e., $\mathcal{F} = \mathcal{S}$, the real-valued set function $L(\mathcal{F})$ is a non-decreasing and submodular function, i.e.,
\begin{equation}
    L(\mathcal{S}\cup \{f\}) -   L(\mathcal{S}) \geq  L(\mathcal{S}'\cup \{f\}) - L(\mathcal{S}')
\end{equation}
where  $\mathcal{S} \subseteq \mathcal{S}'\subseteq \mathcal{T}_S$ and $f \in \mathcal{T}_S\backslash \mathcal{S}'$. 
\end{lemma}
\begin{proof}
See Appendix A.
\end{proof}

\begin{lemma} \label{le.2}
If the edge server caches only MVCs, i.e., $\mathcal{F} = \mathcal{M}$, the real-valued set function $L(\mathcal{F} )$ is a non-decreasing and supermodular function, i.e.,
\begin{equation}
    L(\mathcal{M}\cup \{t\}) -   L(\mathcal{M}) \leq  L(\mathcal{M}'\cup \{t\}) - L(\mathcal{M}')
\end{equation}
where  $\mathcal{M} \subseteq \mathcal{M}'\subseteq \{\mathcal{T}_0 \cup \mathcal{T}_1\}$ and $t \in \{\mathcal{T}_0 \cup \mathcal{T}_1\}\backslash \mathcal{M}'$. 
\end{lemma}
\begin{proof}
See Appendix B.
\end{proof}
Problem~\eqref{obj.caching} is a maximal coverage problem when only SVCs are cached. A greedy algorithm can be applied to obtain a near-optimal solution to the maximal coverage problem~\cite{khuller1999budgeted,shanmugam2013femtocaching}. However, when caching both MVCs and SVCs, the submodularity of function $L(\mathcal{F})$ no longer holds. 
Therefore, we design a heuristic algorithm to cache different types of VCs. The main idea is to successively increase the value of $L(\mathcal{F})$ by updating the sets of VCs cached at the edge server, i.e., $\mathcal{M}$ and $\mathcal{S}$. The proposed content placement algorithm has three parts, as presented in Algs.~\ref{al.caching1} to~\ref{al.caching3}, respectively.  The three parts aim to trade off content delivery delay reduction, by caching SVCs, and cache resource usage minimization, by caching MVCs. In the algorithms, we introduce an index as a basis for placing the VCs in the cache, which is defined by
\[l(\mathcal{C}| \mathcal{F}) =  L(\mathcal{C} \cup \mathcal{F})- L(\mathcal{F}).\] 
The index, $l(\mathcal{C}| \mathcal{F})$, represents the increased probability of delay requirement satisfaction by caching the VCs in set $\mathcal{C}$, given that the VCs in set $\mathcal{F}$ are already cached.
In addition, the overall data size of the VCs in set $\mathcal{F}$ is defined by 
\[w(\mathcal{F}) = \sum_{t \in \mathcal{M}} w_t^\text{M} + \sum_{f \in \mathcal{S}} w_f^\text{S}.\]

\begin{algorithm}[t]
\caption{Content Placement Scheme (Part 1)} 
\label{al.caching1}
\begin{algorithmic}[1]
\STATEx {\% Placing MVCs to fill the cache}
\STATE {Initialize set $\mathcal{M} \leftarrow \emptyset$.}
\WHILE {$w(\mathcal{M}) \leq C$}
\parState {Find MVCs in set $\mathcal{C}(f^*), f^*\in \mathcal{T}_S$, where:}
\[f^* = \text{argmax}_f l(\mathcal{C}(f)| \mathcal{M}).\]
\IF {$w(\mathcal{C}(f^*) \cup \mathcal{M}) \leq C$}
\STATE {$\mathcal{M} \leftarrow \mathcal{C}(f^*) \cup \mathcal{M}$.}
\ELSE
\State {\textbf{Break}}
\ENDIF
\ENDWHILE
\STATE {Return set $\mathcal{M}$.}
\end{algorithmic}
\end{algorithm}

\begin{algorithm}[t]
\caption{Content Placement Scheme (Part 2)} 
\label{al.caching2}
\begin{algorithmic}[1]
\STATEx {\% Replace MVCs with SVCs}
\STATE {Initialize set $\mathcal{D} \leftarrow \emptyset$.}
\STATE {Run Algorithm~\ref{al.caching1}, obtain $\mathcal{M}$, and $\mathcal{F} \leftarrow \mathcal{M}$.}
\WHILE {$L(\mathcal{F}') \geq L(\mathcal{F})$}
\STATE {$\mathcal{F} \leftarrow \mathcal{F}'$}
\parState {Find SVC $f_1^*$ or MVCs $\mathcal{C}(f_1^*)$, where:
\!\! \!\!\[\{f_1^*, f_2^*\} \!\!= \!\! \underset{\{f_1, f_2\}}{\mathrm{argmax}} \!\!  
\begin{cases}
l(f_1| \mathcal{F}\backslash \mathcal{C}(f_2) ), \text{ where } f_1 \notin \mathcal{S},\\
l(\mathcal{C}(f_1)| \mathcal{F}\backslash \mathcal{C}(f_2) )
\end{cases}\]
and $\mathcal{C}(f_2) \subseteq \mathcal{M}$.}
\State {$\mathcal{D} \leftarrow \mathcal{C}(f_2^*)$.}
\WHILE {$w(\{f_1^*\} \cup \{\mathcal{F}\backslash \mathcal{D} \}) \geq C$}
\parState {$\mathcal{D} \leftarrow \mathcal{C}(f_3)$ , where $f_3$ can be obtained as follows:}
\[f_3 =\underset{f}{\mathrm{argmax}} L( \mathcal{F}\backslash \{\mathcal{D}\cup \{\mathcal{C}(f)\}) , \mathcal{C}(f)\subseteq \mathcal{M}. \]
\ENDWHILE
\parState {Let $\mathcal{F}' \leftarrow \{f_1^*\} \cup \{\mathcal{F}\backslash \mathcal{D}$\}, if an SVC was selected in Line 4, or $\mathcal{F}' \leftarrow \{C(f_1^*)\} \cup \{\mathcal{F}\backslash \mathcal{D}\}$, if MVCs were selected in Line 4.}
\ENDWHILE
\STATE {Return set $\mathcal{F}_A \leftarrow \mathcal{F}$ and $\Omega_A = L(\mathcal{F})$.}
\end{algorithmic}
\end{algorithm}

\begin{algorithm}[t]
\caption{Content Placement Scheme (Part 3)} 
\label{al.caching3}
\begin{algorithmic}[1]
\STATEx {\% Fill the cache by SVCs to avoid local optimum}
\STATE {Run Algorithm~\ref{al.caching2}, obtain $\mathcal{F}_A$.}
\STATE {Delete all MVCs in set $\mathcal{F}_A$, and the result is denoted by $\mathcal{F}_B$. }
\WHILE {$w(\mathcal{F}_B)\leq C$}
\parState {Find stereoscopic video $f^*$, where:}
\[f^* = \text{argmax}_f l(f| \mathcal{F}_B).\]
\IF {$w(\{f^*\} \cup \mathcal{F}_B) \leq C$}
\STATE {$\mathcal{F}_B \leftarrow \{f^*\} \cup \mathcal{F}_B$.}
\ELSE
\STATE {\textbf{Break}}
\ENDIF
\ENDWHILE
\IF{$L\{\mathcal{F}_B\} < \Omega_A$}
\STATE {$\mathcal{F} \leftarrow \mathcal{F}_A$.}
\ELSE
\STATE {$\mathcal{F} \leftarrow \mathcal{F}_B$.}
\ENDIF
\STATE {Return set $\mathcal{F}$ as the final caching solution, and update $\mathcal{M}$ and $\mathcal{S}$ accordingly.}
\end{algorithmic}
\end{algorithm}

First, in Alg.~\ref{al.caching1}, MVCs are cached to maximize the probability of delay requirement satisfaction. The algorithm provides an initial caching solution in which the edge server caches as many VCs as possible, since MVCs have smaller data sizes than the corresponding SVCs. Specifically, the algorithm selects MVCs with the highest indexes to increase $L(\mathcal{F})$ greedily. For simplicity, in each iteration, a group of MVCs are cached based on the SVC that they can be stitched and projected to. MVCs with the highest indexes are sequentially cached until the cache is full.
Then, based on the initial caching solution provided by Alg.~\ref{al.caching1}, Alg.~\ref{al.caching2} replaces cached MVCs with SVCs and other MVCs that are not cached. The objective is to further increase the probability of delay requirement satisfaction by reducing the computing delay until no MVCs can be replaced to further improve the probability. At the end of Alg.~\ref{al.caching2}, a local optimum of Problem~\eqref{obj.caching} can be obtained. In order to prevent the algorithm from being trapped into a local optimum, Alg.~\ref{al.caching3} evaluates the probability of delay requirement satisfaction when only SVCs are cached. We select the caching solution, i.e., $\mathcal{F}$, which has the highest probability of delay requirement satisfaction from two cases: jointly caching SVCs and MVCs (i.e., $\mathcal{F}_A$ in Alg.~\ref{al.caching2}) and caching SVCs only (i.e., $\mathcal{F}_B$ in Alg.~\ref{al.caching3}). 

The reason for comparing the two cases is given in the following proposition. Let $\omega \in [0,1]$ denote the percentage of the cache space assigned for caching SVCs, i.e., $\sum_{f\in \mathcal{S}} w_f^S = \omega C$. The probability of delay requirement satisfaction by caching SVCs with size $\omega C$ and MVCs with size $(1-\omega) C$ through Alg.~\ref{al.caching2} is denoted by $L(\omega)$. Let $L^M(1-\omega)$ and $L^S(\omega)$ represent the probabilities of delay requirement satisfaction by caching MVCs and SVCs, respectively, where $L(\omega) = L^M(1-\omega)+ L^S(\omega)$. 
\begin{proposition}\label{prop.1}
The following assumptions are made:
\begin{itemize}
    \item The same type of VCs has the same data size.
    \item The ratio $\omega$ can be any real value between 0 and 1, and $L^M(1-\omega)$, and $L^S(\omega)$ are differentiable. 
\end{itemize}
If the cache is initially filled by MVCs, and the MVCs are gradually replaced by the SVCs, i.e., increasing $\omega$, there are no more than two local maximizers of $\omega$ for Problem~\eqref{obj.caching}. If two local maximizers exist, one of them is either $\omega = 1$ or $\omega = 0$. 
\end{proposition}
\begin{proof}
See Appendix C. 
\end{proof}
As shown in Proposition~\ref{prop.1}, Alg.~\ref{al.caching2} may find a local optimum under mild conditions. Therefore, we need to compare it with the probability of delay requirement satisfaction when $\omega = 1$ in Alg.~\ref{al.caching3}. The output of Alg.~\ref{al.caching3} provides a near-optimal caching solution given the limited cache size.  
\subsection{Parallel Content Placement}
Spatio-temporally dividing VR videos that are several minutes in length may generate thousands of VCs. Although Alg.~\ref{al.caching3} can provide a near-optimal caching solution, its complexity can be significant for caching a large number of VCs.  Observing the independence among VCs in different time segments, we propose a parallel content placement scheme. Specifically, we first divide the whole VC catalog $\mathcal{A} = \{\mathcal{T}_0,\mathcal{T}_1,\mathcal{T}_S\}$ into $G$ VC subsets, in which the $g$-th subset is denoted by $\mathcal{A}_g$. Each subset contains VCs in one or multiple time segments. Then, we iteratively determine the cache size allocated for caching VCs in each subset at the edge server. The allocated cache size for VCs in subset $\mathcal{A}_g$ is denoted by $C_g$. Based on $C_g$, Alg. ~\ref{al.caching3} can be used to place VCs in a parallel manner. The optimization problem for determining the optimal cache sizes for the subsets is 
\begin{subequations}
\label{obj:3}
\begin{align}
 \max_{\{C_g, \forall g\}} & &\sum_{g = 1}^G P_g(C_g)\\
 \text{s.t.}& & \sum_{g = 1}^G C_g \leq C, \label{obj:3_c1}\\
 && C_g \geq 0, \forall s\label{obj:3_c2}
\end{align}
\end{subequations}
where function $P_g(C_g)$ represents the probability of delay requirement satisfaction for caching VCs in subset $\mathcal{A}_g$ with cache size $C_g$. The function is defined by 
\[ P_g(C_g)=\sum_{d\in \mathcal{V}_g} p(d)P(T(d) < H) \]
where $\mathcal{V}_g$ denotes the set of viewpoints to be rendered by the VCs in subset $\mathcal{A}_g$. 
The function is quasi-convex and monotonically increases as cache size $C_g$ increases. While the closed-form formulation of $P_g(C_g)$ is intractable, we can solve Problem~\eqref{obj:3} via an optimization method.

We use the relaxed heavy ball ADMM technique to decouple Problem~\eqref{obj:3} into sub-problems, which maximize the probability of delay requirement satisfaction for caching VCs in individual subsets, i.e., $P_g(C_g), \forall g$. The problem decomposition is presented in Appendix D. The algorithm for placing VCs in multiple subsets is Alg.~\ref{al:full}. In Lines~3 and~4, the auxiliary variables (provided in Appendix D) and cache sizes allocated for subsets are determined iteratively. The cache size for the $g$-th subset, i.e., $C_g$, obtained in the $m$-th iteration is denoted by  $C_g^{m}$. From~\eqref{lag2_}, obtaining the value of $C_g^{m+1}$ requires the knowledge of the gradient of $P_g(C_g^{m+1})$. We estimate the function through a piece-wise linear function $\hat{P}_g(C_g^{m+1})$. In each iteration, when $C_g^{m+1}$ is updated, we use Alg.~\ref{al.caching3} to obtain the probability of delay requirement satisfaction, as shown in Lines 7 to 9. Accordingly,  function $\hat{P}_g(C_g)$ is updated by Eq.~\eqref{eq:update} in Line~10, where $C_g^{-}$ represents the maximum cache size for subset $g$ obtained in a previous iteration which is lower than $C_g^{m+1}$. Note that $C_g^{+}$ represents the minimum cache size for subset $g$ obtained in a previous iteration, which is higher than $C_g^{m+1}$. As the number of iterations increases, the estimated function $\hat{P}_g(C_g^{m+1})$ approaches the actual function ${P}_g(C_g^{m+1})$. For exploring the actual value of ${P}_g(C_g)$ in the approximation, we apply a zero-mean variance-attenuated Gaussian noise on $C^{m+1}_g$ in Line~5. The noise, following a normal distribution, $\mathcal{N}(0, \sigma^2)$, has variance $\sigma^2$ attenuated in each iteration, as shown in Line~11. The algorithm is stopped when the Lagrangian of the problem, as shown in~\eqref{eq:Lag}, converges and the exploration on $\hat{P}_g(C_g)$ completes.

\begin{figure*}[t]
\begin{align}
\hat{P}_g^{m+1}(C_g)=
\begin{cases}
\frac{\hat{P}_g^{m+1}(C_g^{m+1}) - \hat{P}_g^m(C_g^{-})}{C_g^{m+1} - C_g^{-}} (C_g-C_g^{-}) +  \hat{P}_g^m(C_g^{-}), & \text{for }  C_g^{-}\leq C_g < C_g^{m+1}\\
\frac{\hat{P}_g^m(C_g^{+})-\hat{P}_g^{m+1}(C_g^{m+1})}{C_g^{+} - C_g^{m+1} } (C_g^{+} - C_g) + \hat{P}_g^{m+1}(C_g^{m+1}), & \text{for }  C_g^{m+1}\leq C_g < C_g^{+}\\
\hat{P}_g^m(C_g), & \text{Otherwise}. 
\end{cases} \label{eq:update}
\end{align}
\begin{equation}
    R(s_{u,k}, a_{u,k})\! = \!\!\begin{cases}
    \sum_{k' = k+\phi}^{\infty} \!\kappa^{k'-k} f(v;\mathcal{B}_{u,k})\textbf{1}[v\in \mathcal{U}_{u,k}(\mathbf{a}_{u,k-1})], \!\!\!\!& \text{if } a_{u,k} =1;\\
    0, & \text{if } a_{u,k} =0.
    \end{cases} \label{eq.reward}
\end{equation}
\begin{align}
    &\lambda(s_{u,k})\!\! = \!\!R(s_{u,k}, 1)\!\! +  \!\!\kappa^\phi\! \big\{\!\!\!\!\sum_{s_{u,k+\phi}\in \mathcal{Y}^u}\!\!\!\!\!\! P(s_{u,k+\phi}|s_{u,k},1)  \!\!\max_{a_{u, k+\phi}}\!\! \{ Q(s_{u,k+\phi}, a_{u, k+\phi})\}\! -\notag\\
 &\!\!\!\! \!\!\!\!\sum_{s_{u,k+\phi}\in \mathcal{Y}^u} \!\!\!\!\!\!\!P(s_{u,k+\phi}|s_{u,k},0) \!\!  \max_{a_{u, k+\phi}} \!\! \{ Q(s_{u,k+\phi}, a_{u, k+\phi})\}\big\}. \label{eq.wi}
\end{align}
\hrulefill
\end{figure*}

The aforementioned problem decomposition and parallel computing significantly reduce the decision space as compared to solving Problem~\eqref{obj:3} directly. Without parallel content placement, the worst time complexity of Alg.~\ref{al.caching3} is $(|\mathcal{V}|(M_1+M_3) + 2|\mathcal{V}|^2M_2)$, where $M_1$, $M_2$, and $M_3$ are the iteration numbers in Algs.~\ref{al.caching1},~\ref{al.caching2}, and~\ref{al.caching3}, respectively. {With an increase in the number of viewpoints, the time complexity increases quadratically.} By Alg.~\ref{al:full}, the content placement for different video subsets is in parallel. If the popularity functions $\{{P}_g(C_g), \forall g\}$ are known in advance, the number of iterations for running the relaxed heavy ball ADMM algorithm can be low, and the worst time complexity of content placement for a subset becomes $(|\mathcal{V}_g|(M_1+M_3) + 2|\mathcal{V}_g|^2M_2)$.

\begin{algorithm}[t]
    \caption{Parallel Content Placement Scheme} \label{al:full}
    \begin{algorithmic}[1]
    \State Initialize estimation function $\hat{P}^0_g(C_g)$, $m = 0$.
   \WHILE 1
    \State Calculate $C_g^{m+1}\{\hat{P}^k_g(C_g)\}$ by~\eqref{lag2_}.
    \State Update variables in set $\Gamma$ by~\eqref{lag3}-\eqref{lag4}.
    \State  $C_g^{m+1} \leftarrow C_g^{m+1} + \mathcal{N}(0,\sigma^2), \forall g$.
    \FOR {$g \in \mathcal{G}$}
    \parState {Place VCs in subset $g$ given cache size $C_g^{m+1}$ by Alg.~\ref{al.caching3}. }
    \State $ \hat{P}_g^{m+1}(C_g^{m+1}) = P_g(C_g^{m+1})$.
    \ENDFOR
    \State Update $\hat{P}_g^{m+1}(C_g)$ by~\eqref{eq:update}.
    \State $\sigma \leftarrow \alpha_{attn}\sigma$, $m = m+1$.
    \IF {$|{\mathcal{L}(\Gamma)^{m+1}-\mathcal{L}(\Gamma)}|$ and $\sigma$ are less than thresholds}
    \STATE \textbf{Break}.
    \ENDIF
    \ENDWHILE
     \STATE {Return the caching solution with cache capacity $C_g$ for all VC subsets.}
    \end{algorithmic}
  \end{algorithm}

\section{DRL-based Content Delivery}
The content placement scheme proposed in Section~\ref{sec.caching} maximizes the probability of delay requirement satisfaction through caching VCs at the edge server. However, a viewpoint may not be always rendered in time, especially when unpopular videos are requested. Therefore, we propose a content delivery scheduling scheme to solve Problem~\eqref{eq.problem} in this section, given {the caching solution and predefined video quality setting $X_0$}. Specifically, we design a Whittle index (WI) based scheduling policy in restless multi-armed bandit (RMAB) and a DRL-based scheme to evaluate the WI in the scheduling policy. 

\subsection{Whittle Index Formulation}
Problem~\eqref{eq.problem} with stochastic constraint~\eqref{eq.constraint} is an RMAB problem. For $U$ VR headsets, we define $U$ corresponding controlled Markov chains. The state of the Markov chain for VR headset $u$ in time slot $k$ is denoted by $s_{u,k}$, and the state space for VR headset $u$ is denoted by $\mathcal{Y}^u$. State $s_{u,k}$ consists of 1) the historical viewpoint trajectory of VR headset $u$, i.e., $\{v_{u,k^*}, k^* = k-P, \dots, k\}$; 2) the SVCs in the local buffer of VR headset $u$; 3) the request $d_{u,k}$; and 4) the sets of VCs cached at the edge server. 

Denote the binary control variable for the Markov chain by $a_{u,k}$, {where $a_{u,k} = 1$ if the request of headset $u$ is scheduled in time slot $k$, and $a_{u,k} = 0$ otherwise.} Consider the average time length for a state transition as an epoch, which is denoted by ${\phi}$. The average number of time slots in an epoch is $\phi = \bar{T}/(E\times \delta)$, which is utilized to find tractable state transition probabilities. The probability of transiting from state $s_{u,k}$ to $s_{u,k+{\phi}}$ is approximated by $P(s_{u,k+\phi}|s_{u,k}, a_{u,k})$, where $\sum_{s_{u,k+\phi}\in \mathcal{Y}^u} P(s_{u,k+\phi}|s_{u,k}, a_{u,k}) = 1$. The reward of scheduling a request is set as the number of future viewpoints that can be rendered from granting the request, given by~\eqref{eq.reward}.
The discounted reward model is applied in~\eqref{eq.reward}, where $\kappa$ is a discount factor. Problem~\eqref{eq.problem} with constraint~\eqref{eq.constraint} is now an RMAB problem. Whenever a computing unit is available, we select one request from one of the~$U$ VR headsets to schedule with unknown transition probabilities. 

Our method of solving the RMAB problem is based on the WI, which is a heuristic solution with excellent empirical performance~\cite{whittle1988restless}. Specifically, the WI-based method defines a subsidy, i.e., the WI, for each VR headset if its current request cannot be scheduled by the edge server immediately according to the current state of the VR headset.  The edge server can receive a higher long-term reward by scheduling the requests with a higher subsidy, i.e., higher WI. Therefore, the edge server can allocate computing resources to requests by comparing the WIs of requests. The WI of state $s_{u,k}$ can be obtained from~\eqref{eq.wi}.
In~\eqref{eq.wi}, $Q(s_{u,k}, a_{u,k})$ is the state-action value, i.e., the Q value, for state $s_{u,k}$, and action $a_{u,k}$. Eq.~\eqref{eq.wi} presents the WI of VR headset $u$ in time slot $k$. 

In time slot $k$, the edge server collects the WIs $\{\lambda(s_{u,k}), \forall u\}$ from all VR headsets and schedules the request of the VR headset with the highest WI value if a computing unit is available. As shown in Appendix E, the RMAB problem is indexable if the content delivery delay for all requests is the same. Even though the content delivery delay can be different in practice, {the standard deviation of the delay is usually small as compared to the average content delivery delay}\footnote{{The standard deviation of the content delivery delay is bounded by half of the computing delay, which is only a portion of the content delivery delay under mild conditions.}}, and the proposed WI-based solution can still be applied.

\subsection{DRL-based Scheduling}
Since the state-action transition probability is unknown for the RMAB problem, the value of the WI in~\eqref{eq.wi} cannot be obtained straightforwardly. We adopt the DRL approach to approximate the WI value for each VR headset in each time slot. Two neural networks are constructed for each VR headset. One neural network approximates the Q value, and the other neural network approximates the WI value. The weight vectors of the two neural networks corresponding to the Q value and the WI value are denoted by $\boldsymbol{\theta}_Q$ and $\boldsymbol{\theta}_W$, respectively. 

\begin{figure}[t]
		\centering
	  	\includegraphics[width=0.5\textwidth]{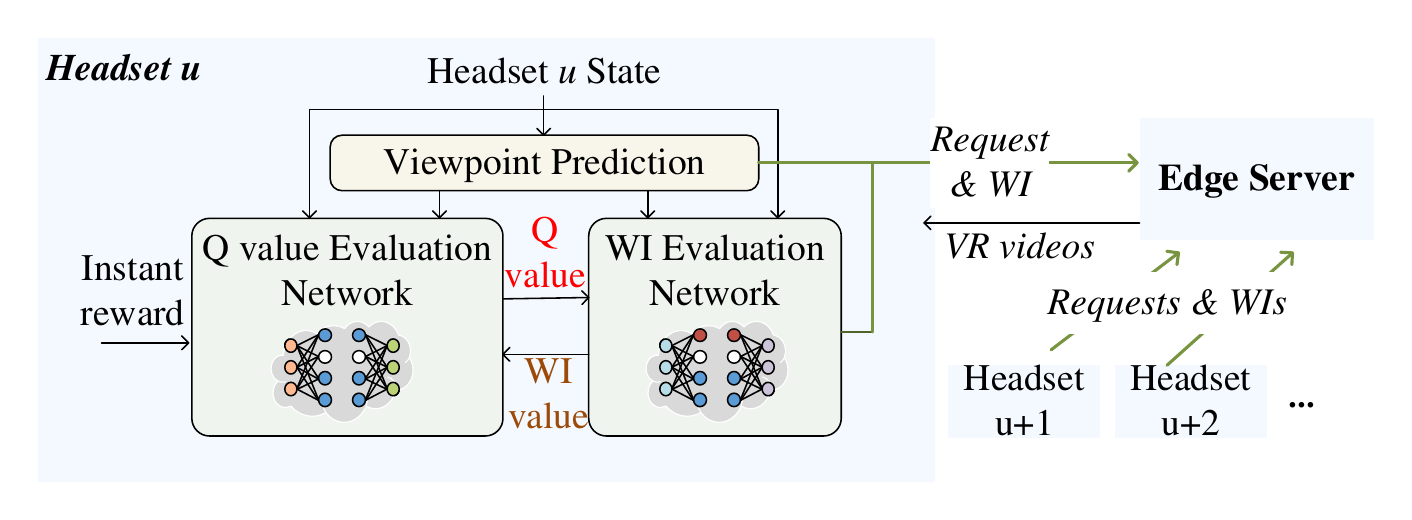}
	  	\caption{DRL-based content delivery scheduling scheme diagram.}\label{fig:5}
\end{figure}

The proposed DRL-based content delivery scheduling scheme is illustrated in Fig.~\ref{fig:5}. At the beginning of time slot $k$, VR headset $u$ observes its state $s_{u,k}$, predicts the future user viewpoints using viewpoint prediction techniques~\cite{Nasrabadi}, and obtains content delivery requests. The Q value and WI evaluation networks approximate the set of Q values and WI values regarding state $s_{u,k}$, denoted by ${Q}(s_{u,k}, a_{u,k}|\boldsymbol{\theta}_Q)$ and ${\lambda}(s_{u,k}, a_{u,k}|\boldsymbol{\theta}_W)$, respectively. With the set of requests, the approximated WI values are sent to the edge server, and the edge server schedules the request with the highest WI when a computing unit is available. Furthermore, at each time slot, the weight vectors $\boldsymbol{\theta}_Q$ and $\boldsymbol{\theta}_W$ are updated according to the reward obtained in content delivery scheduling. In particular, if the request of VR headset $u$ is scheduled, the VR headset observes the reward as given by~\eqref{eq.reward}. Otherwise, we use the approximated WI ${\lambda}(s_{u,k}|\boldsymbol{\theta}_W)$ as the reward, which represents the difference of Q values between two successive states when a request is not satisfied. According to~\cite{AVRACHENKOV2022110186}, weight vector $\boldsymbol{\theta}_Q$ is updated by minimizing the loss function for evaluating the Q value, given by
\begin{align}
    &L(\boldsymbol{\theta}_Q)\! = \!\![{Q}(s_{u,k}, a_{u,k}|\boldsymbol{\theta}_Q) - {Q}(s_{u,k+\phi}, a_{u,k}|\boldsymbol{\theta}_Q)\notag \\
    &\hspace{0.5cm}\lambda(s_{u,k}|\boldsymbol{\theta}_W) (1- a_{u,k}) - R(s_{u,k}, a_{u,k}) a_{u,k}]^2.
\end{align}
Furthermore, we generate a reference WI value based on the outputs of the Q value and WI value neural networks, where
\begin{equation}
    \hat{\lambda}(s_{u,k}) = {\lambda}(s_{u,k}|\boldsymbol{\theta}_W) - \varphi[{Q}(s_{u,k}, 0|\boldsymbol{\theta}_Q) -{Q}(s_{u,k}, 1|\boldsymbol{\theta}_Q)]. \label{eq.lambda}
\end{equation}
In~\eqref{eq.lambda}, $\varphi$ is the step size for approximating the WI value. From~\eqref{eq.lambda}, if the approximation of ${\lambda}(s_{u,k})$ is accurate, ${Q}(s_{u,k}, 0|\boldsymbol{\theta}_Q)$ should be equal to ${Q}(s_{u,k}, 1|\boldsymbol{\theta}_Q)$. Thus, weight vector $\boldsymbol{\theta}_Q$ is updated by minimizing the loss function for evaluating the WI value, given by
\begin{equation}
    L(\boldsymbol{\theta}_W) = [ \varphi({Q}(s_{u,k}, 0|\boldsymbol{\theta}_Q) -{Q}(s_{u,k}, 1|\boldsymbol{\theta}_Q))]^2.
\end{equation}
The proposed DRL-based content delivery scheduling scheme is summarized in Alg.~\ref{al:rmab}, where $\varphi_Q$ and $\varphi_W$ are the learning rates for evaluating weight vectors $\boldsymbol{\theta}_Q$ and $\boldsymbol{\theta}_W$, respectively. At each time slot, a VR headset calculates WI,  which is a float number (4 bytes), in a distributed manner and uploads a WI for content delivery scheduling. Compared to sending the full VR headset state to the edge server for scheduling, the proposed scheme can significantly reduce communication overhead. 

\begin{algorithm}[t]
    \caption{DRL-based Content Delivery Scheduling Scheme}\label{al:rmab}
    \begin{algorithmic}[1]
      \STATE Initialize weight vectors $\theta_W$ and $\theta_Q$ for all VR headsets. $\epsilon = 1$.
      \FOR {time slot $k  = 1, \dots, K$}
      \FOR {VR headset $u = 1:U$ }
      \parState{Request the new SVC and report WI set $\{\lambda(s_{u,k}; \theta_W)\}$.}
      \ENDFOR
      \IF{any computing unit is available}
      \parState{With probability $\epsilon$, randomly schedule a request; otherwise, schedule request $d^*$ from the VR headset with the highest WI.}
      \ENDIF
      \FOR {VR headset $u = 1:U$} 
      \parState{Update the Q value evaluation network by $\theta_Q \leftarrow \theta_Q - \varphi_Q \nabla  L(\theta_Q)$.}
      \parState{Update the WI evaluation network by $\theta_W \leftarrow \theta_W - \varphi_W \nabla  L(\theta_W).$}
      \ENDFOR
      \STATE \textbf{If} $\epsilon > \epsilon_{min}$, $\epsilon \leftarrow \beta_{attn}\epsilon$.
      \ENDFOR
    \end{algorithmic}
  \end{algorithm}
\section{Adaptive Video Quality Adjustment}

The edge caching and content delivery scheduling schemes proposed in Sections V and VI are based on a predefined video quality setting, i.e., a given $X_0$. In this section, we discuss how to optimize video quality in content delivery while meeting frame missing rate requirements. For simplicity, we assume that the data sizes of different MVCs are similar{~\cite{hooft2019tile}}, and the data size of an SVC does not depend on $X_0$. The frame missing rate is jointly determined by the following three factors, as illustrated in Fig.~\ref{fig:delivery}:
\begin{enumerate}
    \item Request waiting time: The request waiting time can be obtained by the time difference between the instant when a content delivery request is generated and the next instant when a request is scheduled by the edge server. The average request waiting time is denoted by $W$.
    \item  Content delivery delay: When a request is scheduled by the edge server, a computing unit is occupied to process and deliver the corresponding SVC. The content delivery delay for request $d_{u,k}$, i.e., $T(d_{u,k})$, is evaluated as in Subsection III.C. 
    \item Request deadline: The request deadline is the time slot when the requested video will be played on the VR headset. In the example shown in Fig.~\ref{fig:delivery}, slot $k'$ is the request deadline of the request generated in slot $k$.
\end{enumerate}

The request waiting time depends on the content delivery scheduling scheme and the number of VR headsets in the system. The request deadline depends on the viewpoint dynamics of the corresponding user and parameter $X_0$ (since $X_0$ impacts the number of viewpoints that can be rendered by the VR headset using the SVCs in the local buffer). Let random variable $N(X_0)$ denote the number of frame missing events between two successive content delivery for a VR headset given $X_0$, with mean denoted by $\bar{N}(X_0)$.
The relation between $\bar{N}(X_0)$ and the video quality setting, i.e., $X_0$, is given by the following lemma:
\begin{lemma}
$\bar{N}(X_0)$ monotonically decreases as $X_0$ increases, under the assumption that the data sizes of MVCs are constant. \label{le.x} 
\end{lemma}
\begin{proof}
See Appendix F. 
\end{proof}
Based on Lemma~\ref{le.x}, parameter $X_0$ can be adjusted tentatively to maximize video quality, while satisfying the frame missing rate requirement. Specifically, a VR headset can monitor the real-time frame missing rate and upload it to the edge server.  If the average frame missing rate is higher than the frame missing rate requirement, $X_0$ should be increased to reduce the frame missing rate by degrading video quality. Otherwise, $X_0$ can be decreased to improve the video quality, until the frame missing rate requirement is no longer satisfied. 

\begin{figure}[t]
		\centering
	  	\includegraphics[width=0.5\textwidth]{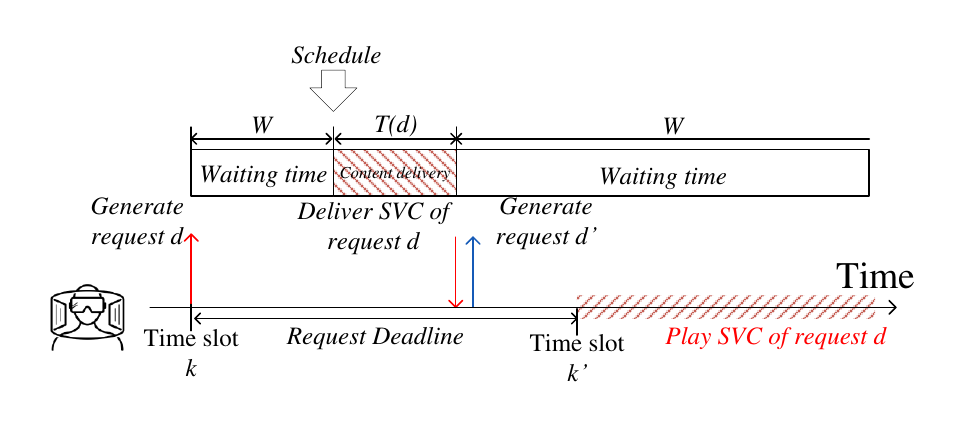}
	  	\caption{Illustration of content delivery processes for a VR headset.}\label{fig:delivery}
\end{figure}

\section{Simulation Results}
\subsection{Parameter Settings}
In the simulation, we use a viewpoint-tracking data set composed of 48 users watching three spherical VR videos~\cite{wu2017dataset}. The details of the VR videos are provided in Table~\ref{table1}. The 360\textdegree~ $\times$ 180\textdegree~ equirectangular video is divided into 24 $\times$ 12 spatial tiles. A user's FoV spans 7 $\times$ 5 tiles. {The length of a time segment, i.e., the playback duration of a VC, is 4 seconds.}  For each time segment, we map the spatial locations of viewpoints into tiles to obtain the probability that request $d$ is generated, i.e., $p(d)$. The data size of MVCs is randomly generated based on a Gaussian distribution, in which the mean is 30 Kbits and the standard deviation is 10 Kbits, and the value of $\alpha$ is 1.3. The content delivery delay threshold  $H$ is 85 ms. The wired transmission rate from the cloud server to the edge server is 700 Mbits/s. In parallel content placement, VCs corresponding to the same time segment, i.e., the playback duration of a VC, form a VC subset, and there are 163 subsets for the three videos. 
In the parallel content placement, the parameters of the relaxed heavy ball ADMM algorithm are set to $\varepsilon_1 = 0.8$, $\varepsilon_2 = 0.85$, and $\rho_1 = \rho_2 = 0.3$.

In the content delivery, there are 10 VR users with VR headsets. {The length of a time slot is 33 ms. } For each user, the viewpoint trajectory is randomly selected from the viewpoint-tracking data set. For processing the content delivery requests from users, the edge server has one computing unit, i.e., $E$ = 1. The transition probabilities of the communication channel between the edge server and users, i.e., $p_H$ and $p_L$, are 0.6 and 0.3, respectively. In the DRL algorithm for content delivery scheduling, the learning rates of the networks for evaluating the Q value and the WI value are 5e-4 and 1e-4, respectively. Neural network structures used to evaluate the Q value and the WI value are presented in Tables~\ref{nn1} and~\ref{nn2}, respectively. At each VR headset, a long short-term memory (LSTM) network with 50 neurons predicts the user's future viewpoint with an accuracy of around 94\%. 
\begin{table}[t]
\caption{Metadata of VR videos~\cite{wu2017dataset}}\label{table1}
\centering{
\begin{tabular}{@{}lllll@{}}
\toprule
No & Video length & Content               & Category & Popularity    \\ \midrule
1  & 2'44"       & Conan360-Sandwich    & Performance & 0.6\\
2  & 3'20"       & Freestyle Skiing      & Sport      & 0.2 \\
3  & 4'48"       & Google Spotlight-Help & Film        & 0.2\\ \bottomrule
\end{tabular}}
\end{table}

\begin{table}[t]
\caption{Network structure for evaluating the WI value}\label{nn1}
\centering{
\begin{tabular}{@{}lll@{}}
\toprule
Layers            & Number of neurons       & Activation Function \\ \midrule
CONV1             & (3$\times$3$\times$ 10) & elu                 \\
POOL1             & (2$\times$2)            & elu                 \\
CONV2             & (3$\times$3$\times$10)  & elu                 \\
LSTM              & 50                      & tanh                \\
Fully Connected 1 & 125                     & elu                 \\
Fully Connected 2 & 64                      & elu                 \\
Fully Connected 3 & 64                      & elu                 \\
Fully Connected 4 & 24                      & elu                 \\
Output            & 1                       & linear              \\ \bottomrule
\end{tabular}}
\end{table}

\begin{table}[t]
\caption{Network structure for evaluating the Q value}\label{nn2}
\centering{
\begin{tabular}{@{}lll@{}}
\toprule
Layers            & Number of neurons       & Activation Function \\ \midrule
CONV1             & (3$\times$3$\times$ 10) & elu                 \\
POOL1             & (2$\times$2)            & elu                 \\
CONV2             & (3$\times$3$\times$10)  & elu                 \\
LSTM              & 50                      & tanh                \\
Fully Connected 1 & 125                     & elu                 \\
Fully Connected 2 & 64                      & elu                 \\
Output            & 1                       & linear              \\ \bottomrule
\end{tabular}}
\end{table}
\begin{figure}[t]
		\centering\hspace*{-0.5cm}
	  	\includegraphics[width=0.45\textwidth]{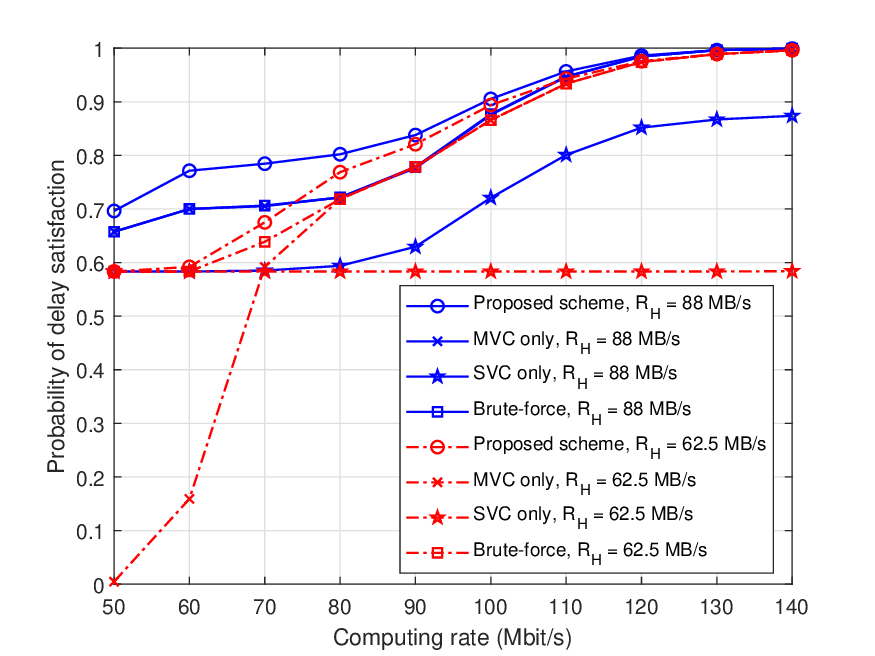}
	  	\caption{The probability of delay requirement satisfaction with various computing rates, {where $R_L = 50$ MB/s}.}\label{sim1}
\end{figure}
\subsection{Performance of the Content Placement Scheme}
We first evaluate the performance of the proposed content placement scheme (i.e., Algs.~\ref{al.caching1}-\ref{al.caching3}) for caching VCs within a VC subset. We compare our proposed scheme with three benchmark schemes: \textit{MVC only} and \textit{SVC only}, in which only MVCs or SVCs are cached, respectively; \textit{cache size search}, which searches the best cache size for caching SVCs, i.e., the best $\omega$, from set \{0, 0.1, \dots, 1\}. The probability of delay requirement satisfaction with various computing rates $\chi$ is shown in Fig.~\ref{sim1}, and the probability of delay requirement satisfaction with various channel transition probabilities, i.e., $p_H$, is shown in Fig.~\ref{sim2}. In comparison with the \textit{MVC-only} and \textit{SVC-only} schemes, the proposed content placement scheme can significantly improve the delay requirement satisfaction probability. This is because the proposed scheme can find the near-optimal caching solution by balancing computing delay and cache resource usage. Furthermore, the proposed scheme outperforms the \textit{cache size search} scheme, which is a computationally intensive scheme based on exhaustive search, {in satisfying the delay requirement}. 
In addition, in both Figs.~\ref{sim1} and~\ref{sim2}, the performance of the \textit{MVC only} scheme improves significantly as the computing rate increases, whereas the performance improvement of the \textit{SVC only} scheme is not significant, especially when the transmission rate $R_H$ is low. This is because the edge server may frequently download MVCs from the cloud server, and the performance bottleneck in the \textit{SVC only} scheme is on the transmission delay for delivering VCs in both wireless and wired links. When the edge server has sufficient computing resources, caching more MVCs can increase the probability of delay requirement satisfaction. 

\begin{figure}[t]
		\centering \hspace*{-0.5cm}
	  	\includegraphics[width=0.45\textwidth]{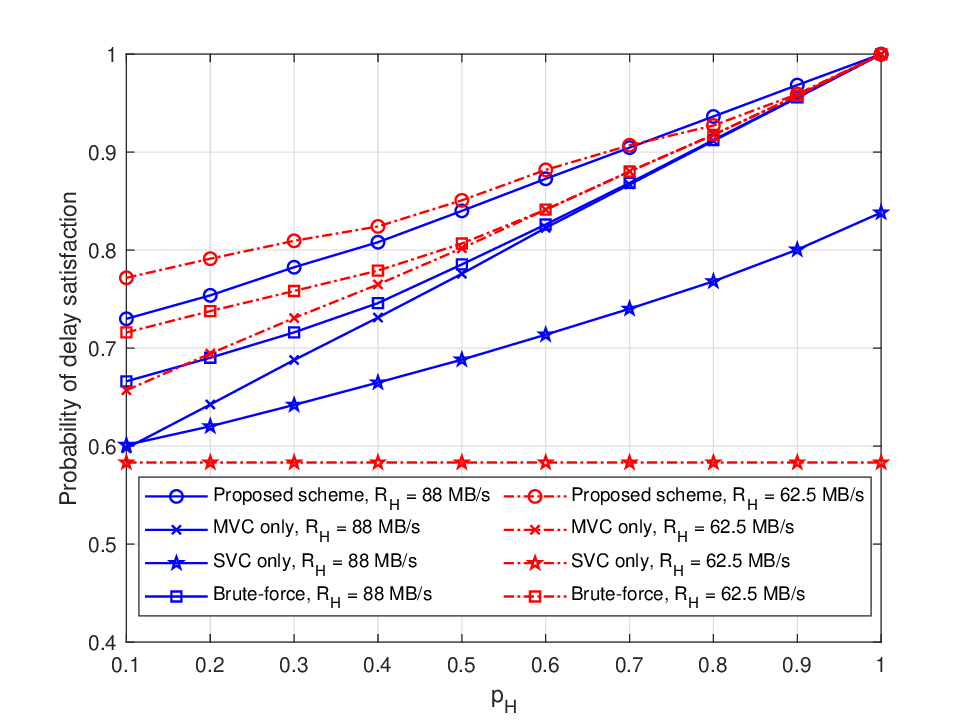}
	  	\caption{The probability of delay requirement satisfaction with various channel transition probabilities, {where $R_L = 50$ MB/s}.}\label{sim2}
\end{figure}

The performance of the parallel content placement scheme (Alg.~\ref{al:full}) for caching VCs associated with multiple time segments is presented in Figs.~\ref{sim3} and~\ref{sim4}, where $R_H = 88$ MB/s and $R_L = 50$ MB/s. Fig.~\ref{sim3} shows the correlation between the popularity of VC subsets, i.e., $\{\sum_{d\in \mathcal{V}_g} p(d), \forall g\}$, and the allocated cache sizes for the subsets, i.e.,  $\{C_g, \forall g\}$, in which all 163 subsets are taken into account. Popular VC subsets are allocated with larger cache sizes than less popular VC subsets. We observe that even if the popularity of different VC subsets is the same, the allocated cache sizes can be different since viewpoint popularity distributions are different among subsets. The correlation between the popularity of VC subsets in two different videos and the allocated cache sizes is shown in Fig.~\ref{sim4}. We use 35 VC subsets for both videos 1 and 3 in Table~\ref{table1} and assign random popularity for the VC subsets. Compared to video 1 (performance video), the popularity and the allocated cache size for a subset in video 3 (film video) are not strongly correlated, especially for the subsets with high popularity. The reason is, compared to video 1, user viewpoints in video 3 are more concentrated.  In such a case, increasing the cache sizes for popular VC subsets may not significantly improve the overall probability of delay requirement satisfaction.
\begin{figure}[t]
		\centering
	  	\includegraphics[width=0.45\textwidth]{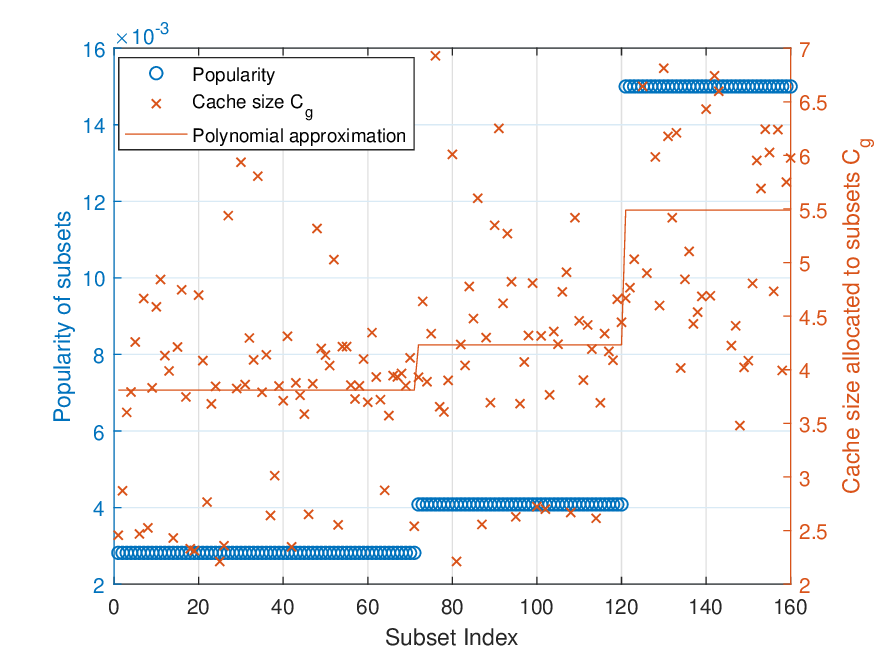}
	  	\caption{The correlation between the popularity of VC subsets, and the allocated cache sizes for VC subsets. }\label{sim3}
\end{figure}

\begin{figure}[t]
		\centering 
	  	\includegraphics[width=0.4\textwidth]{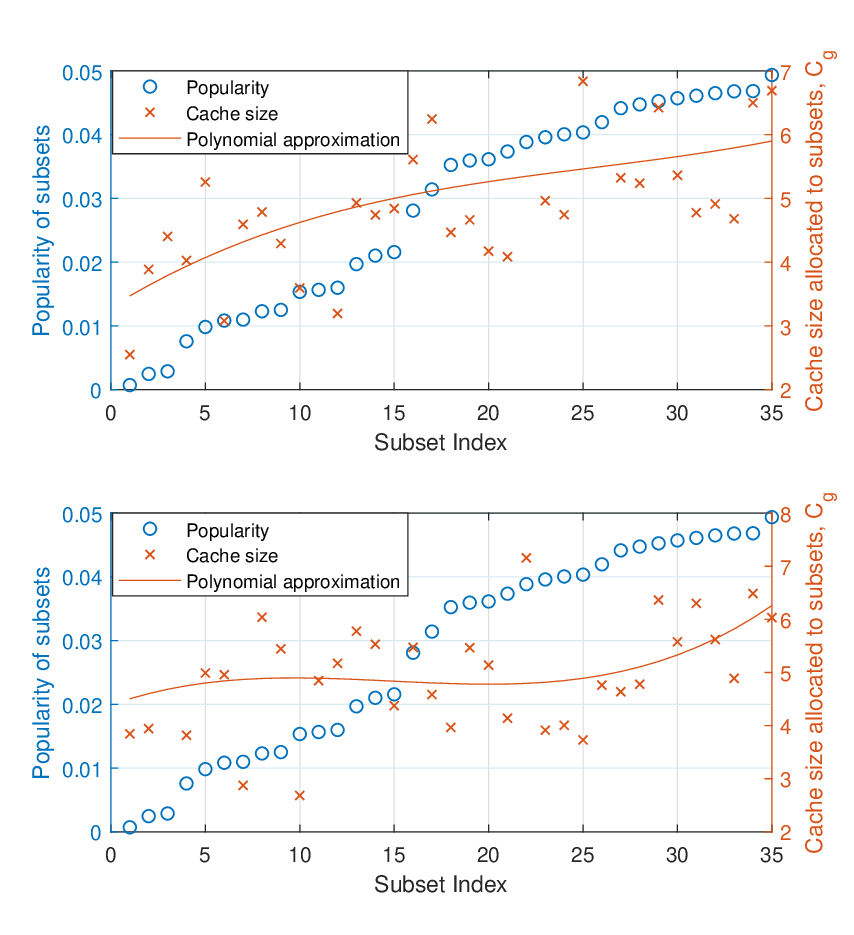}
	  	\caption{The correlation between the popularity of VC subsets in two different videos and the allocated cache sizes for VC subsets. (a) video 1; (b) video 2.}\label{sim4}
\end{figure}

\subsection{Performance of the Content Delivery Scheduling Scheme}
The performance of the proposed WI-based content delivery scheduling scheme is compared with three benchmark schemes: \textit{urgent-request-first (URF)}, in which the edge server always schedules the request for the SVC that will be played in the nearest future first;  \textit{round-robin}, in which the edge server schedules requests in a circular order; and \textit{random}, in which the edge server schedules requests randomly. We generate two different video selection and viewpoint movement profiles based on the viewpoint-tracking data set: one for training the neural networks; another for testing the learning-based scheme performance. For each profile, we generate requests for 2,000 seconds. With the viewpoint movement prediction by LSTM, the performance of the average hit probability for the four schemes is shown in Fig.~\ref{sim5}, in which the test profile is applied. The results in Fig.~\ref{sim5} are the moving average of the hit probability in 1,000 preceding decision epochs (around 8,000 time slots) for each decision epoch. As shown in the figure, the proposed WI-based content delivery scheduling can significantly improve the hit probability. When the transmission rate from the edge server to a headset is low, the proposed WI-based scheme can improve the hit probability by up to 30\% compared to \textit{URF} and \textit{round-robin} schemes. Our proposed WI-based scheme improves performance because, compared to the benchmark schemes, our proposed WI-based scheme evaluates viewpoint and network dynamics in making content delivery scheduling decisions. 
The relation between video quality and average hit probability is shown in Fig.~\ref{sim6}. As the portion of layer-0 MVCs in an SVC, i.e., $X_0$, decreases, the hit probability decreases accordingly, while the video quality is improved. The proportional relation is consistent with Lemma~\ref{le.x}. When implementing content placement and scheduling schemes, parameter $X_0$ should be adapted to achieve the optimal balance between video quality and hit probability.

\begin{figure}[t]
		\centering
	  	\includegraphics[width=0.45\textwidth]{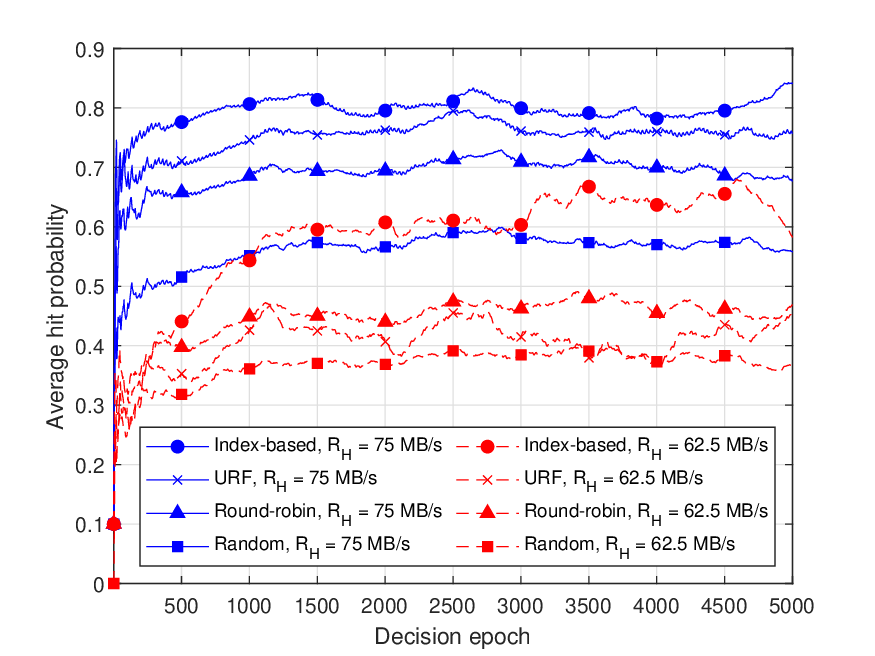}
	  	\caption{The average hit probability with LSTM-based prediction and dynamic channel.}\label{sim5}
\end{figure}

\begin{figure}[t]
		\centering
	  	\includegraphics[width=0.45\textwidth]{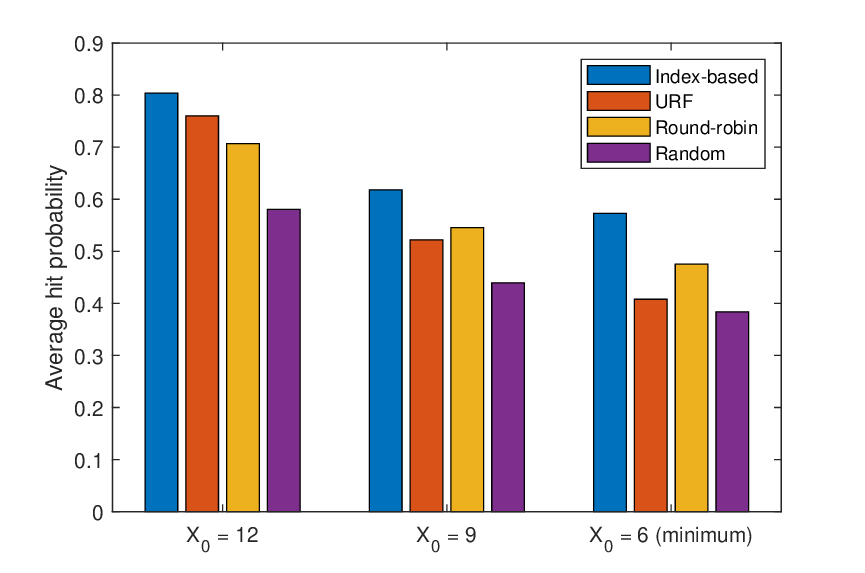}
	  	\caption{The average hit probability versus various video quality settings. }\label{sim6}
\end{figure}
\section{Conclusion}
In this paper, we have studied content caching and delivery schemes for mobile VR video streaming. {Given limited network resources, a scalable content placement scheme has been designed to obtain caching solutions for VR videos by trading off caching and computing resources. In addition, we have developed a novel distributed learning-based content delivery scheduling scheme to coordinate real-time content delivery for multiple VR headsets. The proposed content delivery solution identifies the priority of content delivery among the requests from multiple VR headsets based on the network and user viewpoint dynamics, thereby
maximizing resource utilization.} Simulation results have demonstrated that the proposed solution outperforms benchmark content caching and delivery schemes. In the future, we will investigate efficient computing resource allocation and caching schemes for interactive VR applications, where the dynamics of VR user interactions will be evaluated.


\section*{Appendix A: Proof of Lemma 1}

It can be directly proved by using~(9) and the definition of submodularity. Let $\mathcal{S}_1$ and $\mathcal{S}_2$ be the subsets of $\mathcal{S}$ satisfying $\mathcal{S}_1 \subseteq \mathcal{S}_2$. For any $f \in \mathcal{T}_S\backslash {S}_2$, we have
\[L(\mathcal{S}_1\cup \{f\}) -   L(\mathcal{S}_1)  = L(\mathcal{\mathcal{X}}_1)\]
where 
\[\mathcal{X}_1 = \{\cup_{f'\in \mathcal{S}_1 \cup \{f\}}\mathcal{R}(f')\} \bigcap \{\cup_{f'\in \mathcal{S}_1 }\mathcal{R}(f')\}. \]
Similarly, 
\[\mathcal{X}_2 = \{\cup_{f'\in \mathcal{S}_2 \cup \{f\}}\mathcal{R}(f')\} \bigcap \{\cup_{f'\in \mathcal{S}_2 }\mathcal{R}(f')\} \]
and $L(\mathcal{S}_2\cup \{f\}) -   L(\mathcal{S}_2)  = L(\mathcal{\mathcal{X}}_2)$. Since  $\mathcal{S}_1 \subseteq \mathcal{S}_2$, some viewpoints that can be rendered by the VR headset using SVC $f$ and cannot be satisfied by SVCs in set $\mathcal{S}_1$ may be satisfied by SVCs in set $\mathcal{S}_2$. Therefore, $\mathcal{X}_2\subseteq \mathcal{X}_1$, and $L(\mathcal{S}_1\cup \{f\}) -   L(\mathcal{S}_1)  \geq L(\mathcal{S}_2\cup \{f\}) -   L(\mathcal{S}_2)$.
\section*{Appendix B: Proof of Lemma 2}
It can be proved similarly to Lemma 1. Let $\mathcal{M}_1$ and $\mathcal{M}_2$ be the subsets of $\mathcal{M}$ satisfying $\mathcal{M}_1 \subseteq \mathcal{M}_2$, and $L(\mathcal{M}_1\cup \{t\}) -   L(\mathcal{M}_1)  = L({\mathcal{X}}_1)$. In this case, ${\mathcal{X}}_1$ can be defined by
\[\mathcal{X}_1 = \cup_{f'\in \{f|\mathcal{C}(f)\subseteq \{\mathcal{M}_1\cup \{t\}\}, t\in \mathcal{C}(f)\}}\mathcal{R}(f'). \]
Similarly, let  $L(\mathcal{M}_2\cup \{t\}) -   L(\mathcal{M}_2)  = L(\mathcal{\mathcal{X}}_2)$, where
\[\mathcal{X}_2 = \cup_{f'\in \{f|\mathcal{C}(f)\subseteq \{\mathcal{M}_2\cup \{t\}\}, t\in \mathcal{C}(f)\}}\mathcal{R}(f'). \]
Since $\mathcal{M}_1 \subseteq \mathcal{M}_2$, by MVC $t$, more SVCs can be stitched when MVCs in $\mathcal{M}_2$ are cached as compared to when MVCs in $\mathcal{M}_1$ are cached. Therefore, $\mathcal{X}_1 \subseteq \mathcal{X}_2$, and  $L(\mathcal{M}_1\cup \{t\}) -   L(\mathcal{M}_1)  \leq L(\mathcal{M}_2\cup \{t\}) -   L(\mathcal{M}_2)$.

\section*{Appendix C: Proof of Proposition 1}

Let functions $L^M(\omega|x)$ and $L^S(\omega|x)$ denote the probabilities of delay requirement satisfaction by caching MVCs and SVCs with data size $\omega C$, respectively, given that the same type of videos with data size $x$ has been cached. If both $L^M(1-\omega)$ and $L^S(\omega)$ are differentiable, $L(\omega)$ can be differentiated as follows:
\begin{align*}
    \nabla_\omega L(\omega) &= \nabla_\omega L^S(\omega) + \nabla_\omega [L^M(1) - L^M(\omega| 1-\omega)]\\
    &= \nabla_\omega [L^S(\omega) - L^M(\omega| 1-\omega)]\\
    &= \lim_{\Delta \rightarrow 0} \frac{1}{\Delta} \big\{ L^S(\omega + \Delta) - L^S(\omega)+\\
    & \hspace{1cm} L^M(\omega+ \Delta| 1-\omega- \Delta) -  L^M(\omega| 1-\omega)\big\}\\
    &= \lim_{\Delta \rightarrow 0} \frac{1}{\Delta} \big\{ L^S(\Delta|\omega) - L^M( \Delta|1-\omega-\Delta)\big\}.
\end{align*}
Assuming that all SVCs have the same data size,  $L^S(\Delta|\omega)$ monotonically decreases as $\omega$ increases due to submodularity proved in Lemma 1. Moreover, it can be directly proved that $L^S(\Delta|\omega)$ is convex on $\omega$. This is because, when caching an additional SVC, with fewer SVCs cached at the edge server, viewpoints that can be rendered by VR headsets using the additional SVC are less likely to intersect with viewpoints rendered by the cached SVCs, i.e., 
\[L^S(\Delta|\omega+\delta) - L^S(\Delta|\omega) \geq L^S(\Delta|\omega+2\delta) - L^S(\Delta|\omega+\delta).\]
Similarly, it can be shown that function $L^M( \Delta|1-\omega-\Delta)$ monotonically decreases as $\omega$ increases due to supermodularity proved in Lemma 2, and the function is concave on $\omega$. 

If a local optimal $\omega^*$ exists, $\nabla_\omega L(\omega^*) =0$. In such case, $L^S(\Delta|\omega) = L^M( \Delta|1-\omega-\Delta)$ holds for any $\Delta$. Therefore, as shown in Fig. \ref{fig:proof}, there are three possible cases in finding the local optimal $\omega^*$:
\begin{figure}[t]
		\centering
	  	\includegraphics[width=0.35\textwidth]{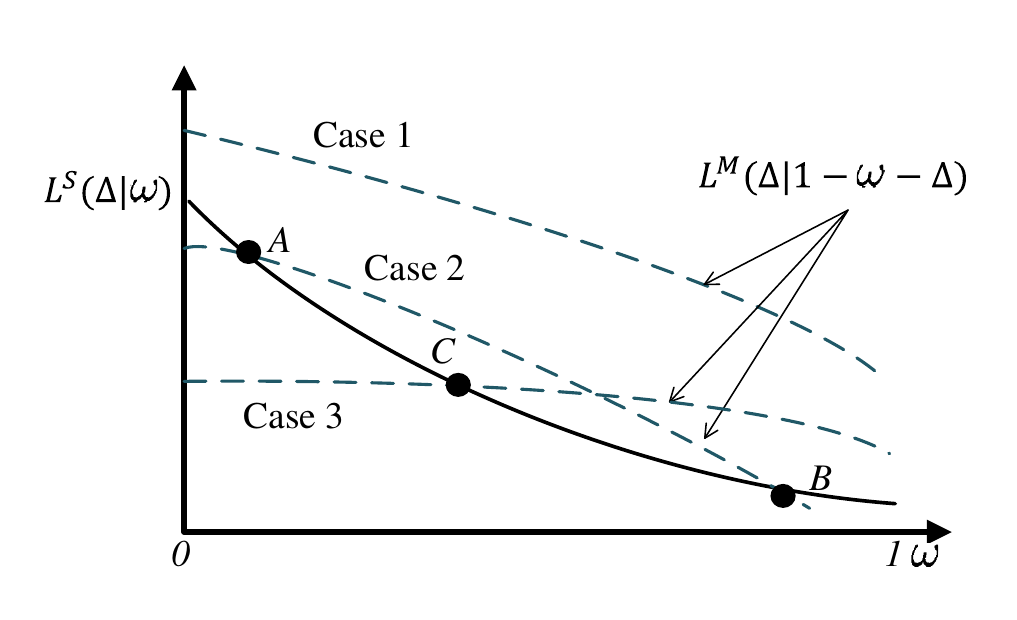}
	  	\caption{Illustration of the proof of Proposition 1.}\label{fig:proof}
\end{figure}
\begin{itemize}
    \item Case 1: No intersection between $L^S(\Delta|\omega)$ and $L^M( \Delta|1-\omega-\Delta)$. Thus, no local optimum between the region $\omega \in (0,1)$. The global maximizer of $L(\omega)$ is located either in $\omega = 0$ or $\omega = 1$.
    \item Case 2: Two intersections between $L^S(\Delta|\omega)$ and $L^M( \Delta|1-\omega-\Delta)$ (i.e., points A and B in Fig. \ref{fig:proof}). Thus, there are two local optimums when $\omega \in (0,1)$, and only one of the local optimums is the local maximizer, denoted by $\omega^*$. If $\omega^*\in (0,1)$ is not the global maximizer, the global maximizer is located at $\omega = 1$. This is because, as $\omega$ increases from zero to the first intersection (i.e., intersection A in Figure \ref{fig:proof}), $\nabla_\omega L(\omega)$ has to be greater than zero. Therefore, the value of $L(\omega)$ at intersection A has to be greater than $L(\omega)$. 
    \item Case 3: Only one intersection between $L^S(\Delta|\omega)$ and $L^M( \Delta|1-\omega-\Delta)$ (i.e., point C in Fig. \ref{fig:proof}). There is at most one local maximizer when $\omega \in (0,1)$, and the other local maximizer is either $\omega = 1$ or $\omega = 0$.
\end{itemize}
In summary, based on the three cases, there are mostly two local maximizers of ratio $\omega$ for Problem~(6). 

\section*{Appendix D: Problem Decomposition for~\eqref{obj:3}}
The first step of parallel content placement is to decouple Problem~\eqref{obj:3} into subproblems for caching the VCs in individual subsets. Introducing auxiliary variables $z_1$ and $\{z_{2,g}, \forall g\}$, we reformulate constraints~\eqref{obj:3_c1} and~\eqref{obj:3_c2} as
\begin{subequations}
\begin{align}
    &\bar{C} - \frac{C_g}{|\mathcal{G}|} - z_1 = 0, \label{obj:3_c11}\\
    &C_g - z_{2, g} = 0, \forall g \label{obj:3_c12}
\end{align}
\end{subequations}
respectively, where $\bar{C}$ represents the mean value of $\{C_g, \forall g\}$. We formulate the augmented Lagrangian of the problem as
\begin{align}
&\mathcal{L}(\Gamma) \!= \!-\sum_{g\in \mathcal{G}}\!\! P_g(C_g) \!+\! u_1(\bar{C}\!-\frac{C_g}{|\mathcal{G}|}\!-z_1) \!+ \!\!\sum_gu_{2,g}(C_g - z_{2,g})\notag \\
&+ \frac{\rho_{1}}{2}({\bar{C}-\frac{C}{|\mathcal{G}|}-z_1})^2 + \frac{\rho_{2}}{2}\sum_g({C_g - z_{2,g}})^2 \label{eq:Lag}
\end{align}
where $u_1$ and $\{ u_{2,g}, \forall g\}$ are the Lagrangian multipliers for Constraints~\eqref{obj:3_c11} and~\eqref{obj:3_c12}, respectively. Parameters $\rho_{1}$ and $\rho_{2}$ are penalty factors. The term $\Gamma$ represents the set of variables, i.e., $\Gamma = \{\{C_g,z_{2,g},u_{2,g}, \forall g\}, z_1,  u_1\}$. We then adopt relaxed heavy ball ADMM~\cite{francca2018nonsmooth} to update the variables in an iterative and parallel manner, where decoupled subproblems for iteration $e+1$ are as follows:
\begin{subequations}
\begin{align}
&C_g^{e+1} \! =\!  \underset{C_g^{e+1}}{\text{argmin}} \{\!-P_g(C_g^{e+1}\!)\!+\! u_1^{e}(\bar{C}^{e}\!-\!C_g^{e}\!+\!C_g^{e+1}\!-\!\frac{C}{|\mathcal{G}|}\!-\!z_1^{e}) \notag\\
& \!+\! u_{2,g}^{e}(C_g^{e+1} \!-\! z_{2,g}) \!+\! \frac{\rho_{1}}{2}(\bar{C}^{e}-C_g^{e}+C_g^{e+1}-\frac{C_g}{|\mathcal{G}|}-z_1^{e})^2 \notag\\
& + \frac{\rho_{2}}{2}(C_g^{e+1} - z_{2,g}^{e})^2\},\label{lag2_}\\
&z_{2,g}^{e+1} = \max \{\varepsilon_1 C_g^{e+1} \!\!+ \!\!(1-\varepsilon_1 )\hat{z}_{2,g}^{e}+\frac{1}{\rho_2}\hat{u}_{2,g}^{e},0\}, \forall g,\label{lag3}\\
&u_1^{e+1} = \hat{u}_1^{e} + \rho_1[\varepsilon_1(\bar{C}^{e+1}-\frac{C}{|\mathcal{G}|}) + (1-\varepsilon_1)\hat{z}_1^{e} - z_1^{e+1}],\\
&u_{2,g}^{e+1} = \hat{u}_{2,g}^{e} +\rho_1[\varepsilon_1 C_g^{e+1} + (1-\varepsilon_1 )\hat{z}_{2,g}^{e}- z_{2,g}^{e+1}], \forall g,\\
&\hat{u}_1^{e+1} = {u}_1^{e+1} + \varepsilon_2({u}_1^{e+1} - \hat{u}_1^{e}),\\
&\hat{u}_{2,g}^{e+1} = {u}_{2,g}^{e+1} + \varepsilon_2({u}_{2,g}^{e+1} - \hat{u}_{2,g}^{e}), \forall g,\\
&\hat{z}_1^{e+1} = {z}_1^{e+1} + \varepsilon_2({z}_1^{e+1} - \hat{z}_1^{e}),\\
&\hat{z}_{2,g}^{e+1} = {z}_{2,g}^{e+1} + \varepsilon_2({z}_{2,g}^{e+1} - \hat{z}_{2,g}^{e}), \forall g \label{lag4}
\end{align}
\end{subequations}
where $0 \leq \varepsilon_1 \leq 1$ and $0 \leq \varepsilon_2 \leq 1$.

\section*{Appendix E: Indexability of the RMAB Problem}
Let $\mathcal{P}(\lambda^*)$ denote the set of state $\mathcal{Y}$ for which the VR headset stays passive, i.e., $a_{u,k} = 0$, according to subsidy $\lambda^*(s_{u,k}), s_{u,k} \in \mathcal{Y}$.
If the RMAB problem is indexable, the following condition should be satisfied:
\begin{itemize}
\item Set $\mathcal{P}(\lambda^*)$ of the corresponding single-armed bandit process increases monotonically from the empty set to the whole state space as subsidy $\lambda^*$ increases from $-\infty$ to $+\infty$~\cite{5605371}.
\end{itemize}
To prove the indexability, we assume that the content delivery delay for all requests is the same, which is denoted by $\phi$.
If the following condition holds, the problem is indexable~\cite{5605371}:
\begin{align}
\frac{d Q(s_{u,k}, 0)}{d \lambda} \geq \frac{d Q(s_{u,k}, 1)}{d \lambda}. \label{eq:index_1}
\end{align}
Define function $B_{\lambda}(s_{u,k}) = \frac{d V(s_{u,k})}{d \lambda} $, and denote the SVC that is scheduled in time slot $k$ by $b$.
The condition in~\eqref{eq:index_1} is equivalent to 
\begin{align}
1+\kappa^{\phi}B_{\lambda}(s_{u,k+\phi}) \geq  \kappa^{\phi}B_{\lambda}(s_{u,k+\phi}\cup \{b\}). \label{eq:index_2}
\end{align}
To prove the above relation, we first assume that $b$ is the last SVC that can be downloaded, i.e., all SVCs have been downloaded in the buffer except $b$. The left-hand side (LHS) of~\eqref{eq:index_2} can be rewritten as
\begin{align}
\textrm{LHS} & = 1\!+\!\kappa^{\phi}\max [1+\kappa^{\phi}B_{\lambda}(s_{u,k+\phi}), \kappa^{\phi}B_{\lambda}(s_{u,k+2\phi} \! \cup \! \{b\})] \notag\\
& \geq 1+\kappa^{2\phi}B_{\lambda}(s_{u,k+2\phi} \cup \{b\}).
\end{align}
The right-hand side (RHS) of~\eqref{eq:index_2} can be rewritten as
\begin{equation}
\textrm{RHS} = \kappa^{\phi}(1+\kappa^{\phi}B_{\lambda}(s_{u,k+2\phi}\cup \{b\})).
\end{equation}
It can be observed that the value of the $\textrm{LHS}$ is greater than the value of the $\textrm{RHS}$ since $\kappa <1$. 

Moreover, we assume that $b$ is the last $h$-th SVCs to be downloaded from the edge until the SVCs are fully downloaded in the buffer, and~\eqref{eq:index_2} holds in such a case. 
In time slot $k'$, where $n'<n$, the last $(h+1)$-th SVC is requested by the VR headset, which is denoted by $b'$. The LHS of~\eqref{eq:index_2} for skipping to download $b'$ in time slot $k'$ can be represented as
\begin{equation}
\textrm{LHS} \!\!= \!\!1+\kappa^{\phi}\max [1+\kappa^{\phi}B_{\lambda}(s_{u,k'+2\phi}), \kappa^{\phi}B_{\lambda}(s_{u,k'+2\phi} \cup \{b^*\})]
\end{equation}
where $b^*$ can be either $b'$ or $b$. Moreover, the RHS of~\eqref{eq:index_2} for downloading $b'$ in time slot $n'$ can be represented as
\begin{equation}
\textrm{RHS} \!\!= \!\!\kappa^{\phi}\!\max [1\!+\kappa^{\phi}\!B_{\lambda}\!(\!s_{u,k'+2\phi}\cup \{b'\}\!),\! \kappa^{\phi}\!B_{\lambda}\!(\!s_{u,k'+2\phi} \cup \{b',\!b\}\!)].
\end{equation}
Since $b$ is the last $h$-th SVCs to be downloaded, the following condition holds:
\begin{equation}
1+\kappa^{\phi}B_{\lambda}(s_{u,k'+2\phi}\cup \{b'\})\geq \kappa^{\phi}B_{\lambda}(s_{u,k'+2\phi} \cup \{b',b\})
\end{equation}
where the RHS is equal to $ \kappa^{\phi}[1+\kappa^{\phi}B_{\lambda}(s_{u,k'+2\phi}\cup \{b'\})]$. If $b^* = b'$, it is straightforward that the $\textrm{LHS}$ is greater than the $\textrm{RHS}$. If $b^* = b$, and the SVC $b'$ is never requested again in the subsequent time slots, state $s_{u,k'+2\phi}\cup \{b'\}$ is equivalent to $s_{u,k'+2\phi}$. The condition that the value of the $\textrm{LHS}$ is not less than the value of $\textrm{RHS}$ still holds. If SVC $b'$ will be requested again in the subsequent time slots, for both state $s_{u,k'+2\phi}\cup \{b\}$ in the LHS and state $s_{u,k'+2\phi}\cup \{b'\}$ in the RHS, there are $h$ SVCs left to be downloaded, and $B_{\lambda}(s_{u,k'+2\phi}\cup \{b\}) = B_{\lambda}(s_{u,k'+2\phi}\cup \{b'\})$. In such a case, the value of the $\textrm{LHS}$ is no less than the value of the $\textrm{RHS}$. Therefore, when the content delivery delay is the same for different requests, the indexability of the problem holds. 

\section*{Appendix F: Proof of Lemma 3}
\begin{figure}[t]
		\centering
	  	\includegraphics[width=0.42\textwidth]{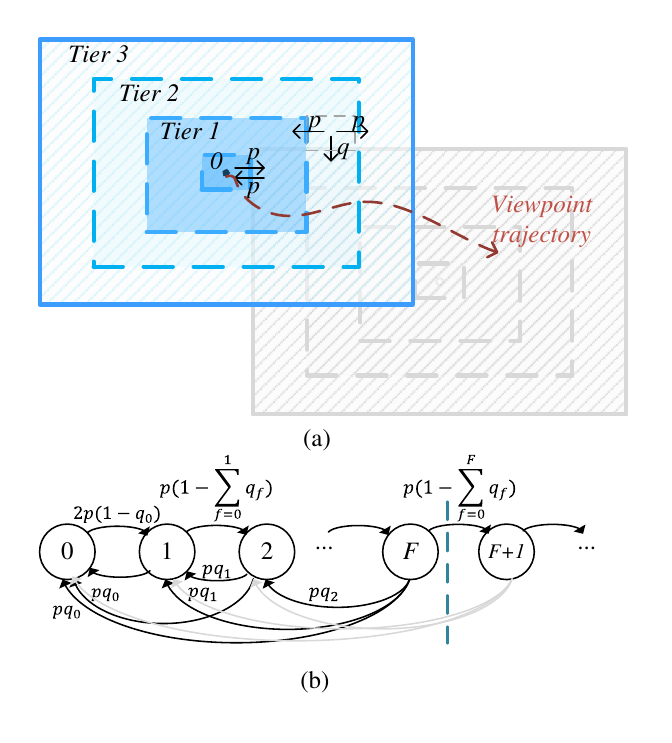}
	  	\caption{(a) Illustration of viewpoint movement, where shadowed areas area are the video spatial area covered by two different videos; (b) Markov chain on viewpoint movement.}\label{fig:markov}
\end{figure}
The distribution of request deadline $N(X_0)$ for a VR headset can be obtained by a Markov chain on viewpoint movement, as shown in Fig. \ref{fig:markov}.
Denote the probability that the spatial location of a viewpoint moves out of a tile along a direction within a time slot by $p$, as shown in Fig. \ref{fig:markov}(a). The viewpoint can move either up, down, left, or right from the current tile. The user's watching behavior and the video they are watching are factors that influence this probability. Furthermore, denote the probability that the spatial location of a viewpoint moves into the tiles and the viewpoint can be rendered by downloading SVC by $q$, where the probability depends on the number of videos downloaded in the buffer. To simplify the notations, we define index $\zeta$ as a tier in the spatial area of an SVC, as shown in Fig. \ref{fig:markov}(a), and the number of tiles in tier $\zeta$ is denoted by $M(\zeta)$. Then, we define parameter $F$ as the maximal tier of tiles where the located viewpoints can be rendered by an SVC. 
Then, we model the viewpoint movement across tires as a Markov chain as shown in Fig. \ref{fig:markov}(b), where state $V(k)$ represents the minimum tier where the viewpoint is located and the tiles are covered by downloaded SVCs in time slot $k$. The transition probabilities between states are defined by
\begin{subequations}
\begin{align}
&\!\!\!\!P(1|0) = 4p(1-q_0);P(0|0) = 1-4p(1-q_0);\\
&\!\!\!\! P(V(k)\!+\!1|V(k))= p(1-\sum_{t = 0}^{\beta}q_{\zeta}), V(k) > 0; \\
&\!\!\!\!P(V(k)\!-\!1|V(k)) \!= \!p(1\!-\!\!\sum_{t = 0}^{\beta-1}q_{\zeta}+4q_{V(k)-1}),\!  V(k) \!> \!0;\\
&\!\!\!\!P(\zeta|V(k)) = 4pq_{\zeta}, 1<\zeta\leq \beta,  V(k) > 0;\\
&\!\!\!\!P(V(k)|V(k))\! =\! 1\!-\!2p(1+\sum_{\zeta = 0}^{\beta-1}q_{\zeta})\!+pq_{V(k)},\! V(k) > 0 \label{eq.trans}
\end{align}
\end{subequations}
where $q_{\zeta} = q M(\zeta)/\sum_{\zeta'=0}^F M(\zeta')$, and $\beta = \min\{V(k),F\}$. The probability mass function of $N(X_0)$ can be obtained from evaluating the number of times that state $V(k)>F$ within $(W+H)$ time slots, given $V(0) = 0$. As shown in~\eqref{eq.trans}, states are more likely to transit to the states with higher values than $F$ when viewpoints move faster, i.e., higher $p$. When $X_0$ increases, $F$ increases since more viewpoints can be rendered by the same SVCs, and $\bar{N}(X_0)$ decreases accordingly.  We can tentatively decrease $X_0$ whenever the frame missing rate is higher than a requirement. 

\bibliographystyle{IEEEtran}
\bibliography{reference}

\end{document}